\pdfoutput=1
\documentclass[11pt,3p]{elsarticle}

\usepackage[T1]{fontenc} 
\usepackage{amsmath} 
\usepackage{amssymb} 

\usepackage{color}
\usepackage[pdftex,colorlinks,allcolors=blue]{hyperref}
\usepackage{enumitem}	
\usepackage{algorithm}	
\usepackage[noend]{algorithmic}  

\newcommand{\bsy}[1]{\boldsymbol{#1}} 

\newcommand{\dbldag}{\raisebox{1pt}{$\sst\dag\!\dag$}} 
\newcommand{\diag}{\mbox{\rm diag}}
\newcommand{\divides}{\,|\,}    
\newcommand{\dividesnot}{\,\notmid\,}    

\newcommand{\eqdef}{\ensuremath{\stackrel{\mbox{\tiny\textsf{def}}}{=}}}  

\newcommand{\Fl}{\ensuremath{F_{\ell}}}

\newcommand{\lowergam}[1]{\ensuremath{\gamma_{\raisebox{-2pt}{$\ssst#1$}}}} 
\DeclareMathOperator{\Lord}{ord_L}				

\newcommand{\notmid}{\mbox{$\hspace{-1.5pt}\not\hspace{2.4pt}\mid\hspace{1.5pt}$}} 

\newcommand{\Psub}[1]{\mbox{$\mathbf{P}_{\! #1}$}}
\newcommand{\regtrademark}{\raisebox{3.5pt}{\ooalign{\hfil\raise.07ex
    \hbox{\scriptsize R}\hfil\crcr\mathhexbox20D}}}
\newcommand{\Rl}{\ensuremath{R_{\ell}}}

\newcommand{\ssst}{\scriptscriptstyle}
\newcommand{\sst}{\scriptstyle}

\newtheorem{theorem}{Theorem}[section]
\newtheorem{lemma}[theorem]{Lemma}
\newtheorem{corollary}[theorem]{Corollary} 

\newdefinition{definition}[theorem]{Definition} 
\newdefinition{example}[theorem]{Example}
\newcommand{\rem}{\noindent\textbf{Remarks.\ }} 
\newproof{proof}{Proof} 

\begin{document}

\begin{frontmatter}
\title{Factoring Perfect Reconstruction Filter Banks into Causal Lifting Matrices: A Diophantine Approach
\footnote{published in \emph{Journal of Computational Algebra} 12 (2024) 100024, DOI: 10.1016/j.jaca.2024.100024}}
\date{October 2024}

\author[1]{Christopher M.\ Brislawn\corref{cor1}}
\ead{cbrislawn@yahoo.com}
\cortext[cor1]{Corresponding author}
\affiliation[1]{organization={Los Alamos National Laboratory}, addressline={P.O. Box 1663},
postcode={Los Alamos}, city={NM}, country={USA}}

\begin{abstract}
The elementary theory of   bivariate linear Diophantine equations over polynomial rings  is used to construct  causal lifting factorizations (elementary matrix decompositions) for causal two-channel FIR perfect reconstruction  transfer matrices and wavelet transforms. The Diophantine approach  generates causal  factorizations satisfying certain polynomial degree-reducing inequalities, enabling a new  factorization strategy  called the \emph{Causal Complementation Algorithm}.  This provides a causal (i.e., polynomial, hence \emph{realizable}) alternative to the noncausal  lifting scheme developed by Daubechies and Sweldens using the Extended Euclidean Algorithm for Laurent polynomials. The new approach replaces the Euclidean Algorithm with Gaussian elimination employing
a slight generalization of polynomial  division that ensures existence and uniqueness of quotients whose remainders satisfy user-specified divisibility constraints.  The Causal Complementation Algorithm is shown to be more general than the causal version of the Euclidean Algorithm approach by generating additional causal lifting factorizations beyond those obtainable using the polynomial Euclidean Algorithm. 
\end{abstract}

\begin{keyword}
Causality, causal complementation, digital signal processing, Diophantine equation, elementary matrix decomposition, Euclidean Algorithm, filter bank, lifting factorization, polyphase matrix, wavelet transform
\MSC[2020]{13P25, 15A54, 42C40, 65T60, 94A29}
\end{keyword}
\end{frontmatter}

%
\section{Introduction}\label{sec:Intro}
Multirate filter banks are the digital signal processing constituents of wavelet transforms.
A \emph{discrete wavelet transform} (DWT) is a cascade of 
filter banks that decomposes discrete-time signals into time-frequency components, such as  the exponentially scaled Mallat  decomposition. 
Under suitable conditions DWTs correspond 
to analog signal representations called \emph{multiresolution analyses}~\cite{Mallat89c,Daub92,Meyer93,Mallat99}; such  decompositions  offer a wealth of  data representations featuring joint time-frequency localization and fast digital implementations.  
For one measure of their success, 
the number of U.S.\ patents containing the term ``wavelet[s]'' is now in the thousands. For more examples of success, \cite[\S II]{Bris:13b:TIT}  surveys  multirate filter banks in digital communication coding standards.

Figure~\ref{fig:FB} depicts the Z-transform (complexified frequency domain) representation of a two-channel multirate  filter bank with input $X(z)\eqdef\sum_i x(i)z^{-i}$~\cite{Daub92,Vaid93,VettKov95,StrNgu96,Mallat99}. (The author recommends references~\cite{Vaid93,StrNgu96} to readers who are less familiar with, and seek a guide to, the engineering language and concepts used in this subject area.)
The downarrows represent 2:1 subsampling, which halves the sampling rate of the analysis-filtered output from $H_0$ and $H_1$. Uparrows restore the original sampling rate to subbands $Y_0$ and $Y_1$ by inserting zeros prior to synthesis filtering with $G_0$ and $G_1$, then summing to get the reconstructed signal, $\widehat{X}(z)$.
Figure~\ref{fig:FB} is  a \emph{perfect reconstruction} (PR) filter bank if the transfer function $\widehat{X}(z)/X(z)$ is a monomial, i.e., a constant multiple of a delay.   

For suitably defined \emph{polyphase transfer matrices} $ \mathbf{H}(z)$ and $\mathbf{G}(z)$ (denoted in boldface) the system in Figure~\ref{fig:FB} is mathematically equivalent to the \emph{polyphase-with-delay} (PWD)  representation  in Figure~\ref{fig:PWD}~\cite{Vaid93,StrNgu96}.  
The polyphase analysis  matrix, $\mathbf{H}(z)$, is 
the frequency-domain representation of a  linear translation-invariant operator acting boundedly on 
vector-valued discrete-time signals in 
$\ell^2\bigl(\mathbb{Z},\,\mathbb{C}^2\bigr)$.
By Cramer's Rule, a Laurent polynomial matrix $\mathbf{H}(z)$ is the polyphase matrix of a finite impulse response (FIR) PR filter bank with FIR inverse if and only if, for some  gain  $\hat{a}\neq 0$ and delay $\hat{d}\in\mathbb{Z}$, it satisfies
\begin{equation}\label{PRFB}
|\mathbf{H}(z)| \eqdef \det\mathbf{H}(z) = \hat{a}z^{-\hat{d}}.
\end{equation}
In this paper all filter banks are assumed to be FIR systems  satisfying~\eqref{PRFB}.
\begin{figure}[t]
  \begin{center}
    \includegraphics[page=2]{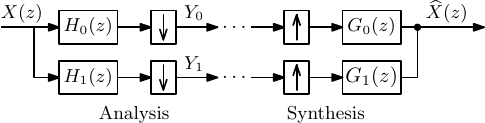}
    \caption{A two-channel multirate analysis/synthesis filter bank.}
    \label{fig:FB}
  \end{center}
\end{figure}
\begin{figure}[t]
  \begin{center}
    \includegraphics[page=3]{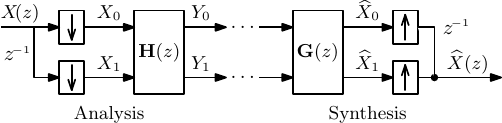}
    \caption{The  polyphase-with-delay (PWD) filter bank representation.}
    \label{fig:PWD}
  \end{center}
  \vspace*{-1em}
\end{figure}

\subsection{Background and Relation to Other Work}\label{sec:Intro:Background}
Many  structures for fast, customizable implementations of PR filter banks   decompose the transfer matrices 
$\mathbf{H}(z)$ and $\mathbf{G}(z)$ into cascades (matrix products) of simpler  building blocks. Examples include  decompositions for particular classes such as paraunitary or linear phase filter banks~\cite{TranQueirozNguyen:00:Linear-phase-perfect-reconstruction,GaoNguyenStrang:01:M-channel-PUFBs,OrainTranHellerNgu:01:paraunitary-linear-phase,GanMaNguyenEtal:02:on-completeness,OrainTranNgu:03:regular-LPFBs,GanMa:TCAS-II-04:simplified-lattice,GanMa:TSP-04:simplified-order-one,MakMuthRed:04:Eigenstructure-approach,XuMakur:09:Arbitrary-Length-LPPRFB} 
and  low-complexity structures like cosine-modulated filter banks~\cite{Vaid93,StrNgu96,PainterSpanias:00:Perceptual-Coding,KhaTuanNgu:ICASSP07:SDP-Cos-Mod-FB,KhaTuanNguyen:09:Cos-Mod-FBs}.
The cascade structures studied in this paper are \emph{lifting factorizations}~\cite{Sweldens96,Sweldens:98:SIAM-lifting-scheme,DaubSwel98} of $\mathbf{H}(z)$ and $\mathbf{G}(z)$ into elementary  (``lifting'') matrices $\mathbf{S}(z)$ of the form 
\begin{align}\label{def_lambda_upsilon}
\mathbf{S}(z) = 
\begin{bmatrix}
1 & 0\\
S(z) & 1\vspace{-1pt}
\end{bmatrix}\text{ or\ \ }
\mathbf{S}(z) = 
\begin{bmatrix}
1 & S(z)\\
0 & 1\vspace{-2pt}
\end{bmatrix}.
\end{align}
%
Lifting  figures prominently in image communication standards like the ISO/IEC JPEG~2000 image coding standards~\cite{ISO_15444_1,ISO_15444_2,TaubMarc02,BrisQuirk03,AcharyaTsai:04:JPEG2000-Standard,Lee:05:J2K-Retrospective}  and  CCSDS Recommendation 122.0 for Space Data System Standards~\cite{CCSDS_122.0:2005,YehArmKielyEtal:AC-05:CCSDS-image-comp}.

The   lifting scheme of Daubechies and Sweldens~\cite{DaubSwel98}  factors   \emph{unimodular} polyphase matrices (matrices of determinant~1) using   computational byproducts of the {Extended Euclidean Algorithm} (EEA)~\cite{MacLaneBirkhoff67,BachShallit:96:Algo-Number-Theory,Shoup:05:Comp-Number-Theory,GathenGerhard:13:Modern-Computer-Algebra} for computing greatest common divisors (gcds) over the  Laurent polynomials  in $z$ and $z^{-1}$, denoted  $\mathbb{C}[z,z^{-1}]$.   The term ``lifting'' seems to  originate with Sweldens~\cite{Sweldens96}, although the result now known as the Lifting Theorem had  appeared earlier in the dissertation of  Herley~\cite{Herley93,VetterliHerley:92:Wavelets-filter-banks}. The EEA was not used in~\cite{Sweldens96}, which used lifting to \emph{synthesize} high-order wavelets  rather than for decomposing a given wavelet transform. The Euclidean Algorithm does appear, however, in  Daubechies' proof of Bezout's Theorem for polynomials~\cite[Theorem~6.1.1]{Daub92}. Daubechies credits Y.~Meyer \cite[Chapter~6, endnote~2]{Daub92} for suggesting the use of Bezout's Theorem to factor trigonometric polynomials when constructing smooth wavelets, so  Meyer deserves some credit for the closely related use of the Euclidean Algorithm in  lifting factorization.

Lifting was anticipated  by others, including Tolhuizen et~al.~\cite{TolHollKal:95:realiz-biorthog-M-D}, who studied  existence of elementary matrix decompositions for \emph{multivariate} matrix polynomials. Cohn~\cite{Cohn:66:structure-ring} constructed an example of a $2\times 2$ bivariate matrix polynomial with no elementary matrix decompositions, but Suslin~\cite{Suslin:77:structure-special} subsequently proved  that elementary matrix decompositions exist for all multivariate matrix polynomials of size $3\times 3$ or greater. A constructive proof  using Gr\"obner bases was found by Park and Woodburn~\cite{ParkWoodburn:95:algorithmic-proof}.   
The univariate  challenge  is not  the existence of elementary matrix decompositions; as noted in~\cite{TolHollKal:95:realiz-biorthog-M-D}, univariate  polynomial (resp., Laurent polynomial) matrices can always be factored into elementary matrices  using the polynomial (resp., Laurent polynomial) EEA.  
The univariate  challenge is, instead, to define and systematically generate ``good'' lifting decompositions well-suited for computational applications.

Since the appearance of~\cite{DaubSwel98}, some work  on  lifting has studied  two-channel filter banks \cite{AdamsKossentini:00:Reversible-wavelet-transforms,MaslenAbbott:00:Automation-lifting-factorisation,ShuiBaoEtal:02:Two-channel-adaptive-biorthogonal,AdamsWard:03:Symmetric-extension-compatible,LiaoMandalEtal:04:Efficient-architectures-lifting-based}  but much  has focused on  $M>2$ channels \cite{Tran:02:M-channel-linear-phase,Tran:02:Rational-LPPRFBs,ChenAmaratun:03:M-channel-lifting-factorization,OrainTranNgu:03:regular-LPFBs,ShuiBao:04:M-band-biorthogonal-interpolating,ChenOrainAmara:05:Dyadic-based-factorizations,IwamuraTanakaIkehara:07:Efficient-Lifting,TanIkeNgu:08:LPFB-Lattice-Structure,SuzukiIkeharaNguyen:12:Generalized-Block-Lifting} and related methods like linear predictive transform coding \cite{WengChenVaid:10:General-Triang-Decomp,WengVaid:12:GTD-Optimizing-PRFBs}.
The limitation of mathematical technique in~\cite{DaubSwel98} (and in most of the  literature since~\cite{DaubSwel98}) to linear algebra and the Euclidean Algorithm strikes the author as unduly restrictive.   Herley's dissertation~\cite{Herley93,VetterliHerley:92:Wavelets-filter-banks} noted a connection to  Diophantine equations, but that idea was not followed up in subsequent papers, and the use of abstract algebra in filter bank theory has remained largely unexplored with only occasional exceptions~\cite{FooteMirchandEtal:00:Wreath-Product-Group,MirchandFooteEtal:00:Wreath-Product-Group,FooteMirchandEtal:04:Two-Dimensional-Wreath-Product,Park:04:Symbolic-computation-signal,LebrunSelesnic:04:Grobner-bases-wavelet,DuBhosriFrazho:10:FB-commutant-lifting,HurParkZheng:14:Multi-D-Wavelet}.  

We shall focus on gaining a deeper understanding of the two-channel case, which has  yielded  significant applications to date such as  digital communications coding and which informs our understanding of the more difficult $M$-channel case.  The two-channel case has also  proven amenable to  nonlinear algebraic methods.
After  serving on the JPEG standards committee (ISO/IEC~JTC1/SC29/WG1)  the author  developed a group-theoretic approach to  lifting  for two-channel unimodular linear phase FIR PR filter banks~\cite{BrisWohl06,Bris:10:GLS-I,Bris:10b:GLS-II}.
Surprisingly, it was shown  that  factoring  linear phase transfer matrices using linear phase lifting filters produces  elementary matrix decompositions that are  \emph{unique} within  ``universes'' of   factorizations that the author called \emph{group lifting structures}.  
These uniqueness results were used~\cite{Bris:13b:TIT} to characterize the group of unimodular  whole-sample symmetric (WS, or odd-length linear phase) filter banks  up to isomorphism. 
It was also shown that the class of unimodular half-sample symmetric (HS, or even-length linear phase) filter banks, which  is \emph{not}  a group, can nonetheless be partitioned into \emph{cosets} of similar groups.  An overview  is given in~\cite{Bris:13:FFT}. 

The present paper is the author's  response to  Daubechies and Sweldens~\cite{DaubSwel98}, who sacrificed causality to exploit nonuniqueness of \emph{Laurent} polynomial division.  
A time-dependent  operator is called \emph{causal} if its output at time $t$ depends only on input data received at times  prior to  $t$; i.e., the operator does not need to ``see into the future'' to compute its output.
A discrete-time FIR filter or filter bank is causal if its  transfer function is \emph{polynomial} in $z^{-1}$ since terms with positive powers of $z$ would indicate a dependence on ``future'' inputs. The transfer functions of  noncausal systems are given by \emph{Laurent} polynomials in both $z$ and $z^{-1}$. Causality is physically necessary for ``realizing'' engineering systems that operate continuously on streaming input, but it also has advantages in computing contexts that enjoy random access to preloaded memory. 

A fundamental limitation of modern High Performance Computing platforms is the latency incurred while writing and reading data to and from the huge amounts of memory provided on such machines. In many cases memory I/O dominates computing performance over and above the runtime cost of the actual calculations being performed on the CPUs. Causal realizations of mathematical operators ensure that the input data need only be read into ``near'' memory (buffers or cache) once, as opposed to a noncausal realization that accesses the same data in ``far'' memory multiple times as both ``future'' and ``past'' samples. Causality also ensures that the memory holding the input data is effectively ``freed up'' once it has been read, allowing the algorithm to be performed \emph{in situ} by writing the latest output samples into the freed memory that previously held the latest inputs. This is a significant advantage for applications with extremely large input arrays since it eliminates the need to allocate memory for an equally large target array to hold the computation's output.

Converting noncausal unimodular factorizations  into ``equivalent'' minimal causal realizations suitable for  implementation is not  straightforward, so sacrificing  realizability as Daubechies and Sweldens have done is a big price to pay for factorization options.  It begs the question of just how many factorizations  Laurent division  creates; there appears to be no systematic way to enumerate   
``all possible'' clever  uses of Laurent division.  It also highlights the lack of a  definition in~\cite{DaubSwel98} of what, exactly, makes an elementary matrix decomposition  a \emph{lifting} factorization.  E.g., what is the difference between the ``nice''  factorizations  in~\cite{DaubSwel98} and  pathological factorizations like~\cite[Proposition~1 and Example~1]{Bris:10:GLS-I} or~\cite[Example~1]{Bris:13b:TIT}?

Given that the goal is to obtain elementary matrix decompositions of FIR transfer matrices, it seems natural to approach the  problem via elementary row and column operations.
By taking an algebraic perspective  we  develop lifting  from basic properties of linear Diophantine equations  (LDEs) over polynomial rings.  We show that much of lifting factorization, including the Lifting Theorem~\cite{Herley93,VetterliHerley:92:Wavelets-filter-banks,Sweldens96}, follows from  basic factorization theory   in commutative rings and does not even involve polynomials per se.  One is led naturally from  abstract algebraic considerations to issues that really require polynomials and causality,  such as degree inequalities and uniqueness results.  This  leads to a  new lifting  strategy based on elementary row and column operations, demonstrated by example in this paper and developed formally in work in progress~\cite{Bris:24:Causal-Complementation-Algorithm}
that we call the \emph{Causal Complementation Algorithm} (CCA).

\subsection{Degree-Reducing Causal Complements and Degree-Lifting Factorizations}\label{sec:Intro:Deg}
Lifting factorization of a transfer matrix $\mathbf{H}(z)$ involves factoring off elementary matrices~\eqref{def_lambda_upsilon}.
E.g.,  left-factorization of an upper-triangular elementary matrix reduces  row~0 (the top row) of  $\mathbf{H}$; suppressing the $z$ for clarity, 
\begin{align}\label{lifting_factorization}
\mathbf{H}(z) 
\eqdef
\begin{bmatrix}
E_0 &E_1\\
F_0 & F_1\vspace{-2pt}
\end{bmatrix}
=
\begin{bmatrix}
1 & S\\
0 & 1\vspace{-2pt}
\end{bmatrix}
\hspace{-3pt}
\begin{bmatrix}
R_0 & R_1\\
F_0 & F_1\vspace{-2pt}
\end{bmatrix}.
\end{align}
To get~\eqref{lifting_factorization}, pick a column index $j\in\{0,1\}$ and divide the pivot $F_j$ into $E_j$  using polynomial  division to get $E_j=F_jS + R_j$, where  remainder $R_j$ satisfies the degree-reducing condition  $\deg(R_j) < \deg(F_j).$ The corresponding remainder in column $j'\eqdef1-j$ of~\eqref{lifting_factorization} for the same lifting filter (or quotient) $S$ is \emph{defined} to be   $R_{j'}\eqdef E_{j'} - F_{j'}S$, so that 
\[  (R_0,R_1) = (E_0,E_1) - S(F_0,F_1), \]
a Gaussian elimination  reduction of row~0. There are  two such lifting factorizations depending on the column in which we perform the polynomial division.   Analogous  left-factorizations of \emph{lower}-triangular elementary matrices correspond to Gaussian elimination  reductions of row~1 in $\mathbf{H}$ while right-factorizations  correspond to column reductions.

Given a  filter $H_i(z)$,  $i\in\{0,1\}$, Herley and Vetterli~\cite{Herley93,VetterliHerley:92:Wavelets-filter-banks} call  a second filter  $H_{i'}(z)$, where  $i'\eqdef 1-i$, a  \emph{complementary filter} if $\{H_0(z),H_1(z)\}$ is a  PR filter bank.  FIR filters $H_0$ and $H_1$ are complementary  if and only if their polyphase matrix satisfies~\eqref{PRFB}, motivating the following polyphase notion of complementary filters in the causal case (i.e., for polynomials in $z^{-1}$).
\begin{definition}\label{defn:CausalComplement}
Let   $F_0,\,F_1$ be causal FIR polyphase filters satisfying $\gcd(F_0,F_1) =  z^{-d}$.
Given  constant $\hat{a}\neq 0$ and integer $\hat{d}\geq d$,
an ordered pair of causal filters $(R_0,R_1)$  is a \emph{causal complement to  $(F_0,F_1)$ for inhomogeneity  $\hat{a}z^{-\hat{d}}$}  if it satisfies the linear Diophantine polynomial equation
\begin{align}\label{complement}
R_0(z)F_1(z) - R_1(z)F_0(z) = \hat{a}z^{-\hat{d}}.
\end{align}
%
We say a causal complement  $(R_0,R_1)$  is  \emph{degree-reducing in} $\Fl,\, \ell\in\{0,1\}$,  if  
\begin{align}\label{deg-reducing}
\deg(\Rl) < \deg(\Fl) - \deg\gcd(F_0,F_1).
\end{align}
\end{definition}

\begin{figure}[t]
  \begin{center}
    \includegraphics[page=1]{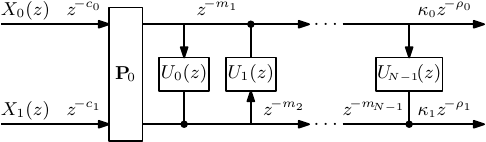}
    \caption{Standard causal lifting form for a FIR PR filter bank.  In this example the  initial lifting matrix, $\mathbf{U}_0(z)$, corresponding to the left-most lifting filter, $U_0(z)$, is lower-triangular, and the number of lifting steps,  $N$, is odd.}
    \label{fig:StdCausalForm}
  \end{center}
\vspace{-1em}
\end{figure}
\rem   In the language of Definition~\ref{defn:CausalComplement}, the CCA constructs factorizations of the form~\eqref{lifting_factorization} by using polynomial division to compute degree-reducing causal complements $(R_0,R_1)$  for inhomogeneity  $\hat{a}z^{-\hat{d}}=|\mathbf{H}(z)|$ without using the EEA.  The reason for the  correction term $\deg\gcd(F_0,F_1)$ in~\eqref{deg-reducing} is explained in Section~\ref{sec:Polynomial:SizeReducing} following Definition~\ref{defn:SizeReducing}; existence and uniqueness of degree-reducing causal complements are addressed by Theorem~\ref{thm:LDDRT}.  

In the course of factorization, the degree-reducing property~\eqref{deg-reducing} eventually drives one of the remainders $R_i$ to zero, causing factorization to terminate. The result can then be put into \emph{standard causal  lifting form} (cf.~Figure~\ref{fig:StdCausalForm}),
%
\begin{align}\label{std_causal_form}
\mathbf{H}(z) =
\diag(\kappa_0 z^{-\rho_0},\kappa_1 z^{-\rho_1})\mathbf{U}_{\! N-1}(z)\bsy{\Lambda}_{N-1}(z)\cdots
\mathbf{U}_1(z)\bsy{\Lambda}_{1}(z)\mathbf{U}_0(z)\Psub{0}\,\diag(z^{-c_0},z^{-c_1}).
\end{align}
The causal matrices $\mathbf{U}_n(z)$ in~\eqref{std_causal_form} are alternating upper- and lower-triangular   lifting matrices~\eqref{def_lambda_upsilon} with causal lifting filters $U_n(z)$.  The matrices $\bsy{\Lambda}_n(z)$ are diagonal delay matrices. They have  a single delay factor $z^{-m_n},$ $m_n\geq 0,$ in the upper channel, 
$\bsy{\Lambda}_n(z) = \diag(z^{-m_n},1)$, or (resp.) the lower channel, $\diag(1,z^{-m_n})$,
if and only if $\mathbf{U}_n(z)$ is upper-triangular (resp., lower-triangular).
\Psub{0} is either the identity, $\mathbf I$, or the \emph{swap matrix}, 
\begin{align}\label{swap_matrix}
\mathbf{J} &\eqdef
\begin{bmatrix}
0 & 1\\
1 & 0\vspace{-1pt}
\end{bmatrix},\quad
\mathbf{J}^{-1} = \mathbf{J}.
\end{align}
Using the CCA, every causal FIR PR  filter bank has a  factorization (many, in fact) in standard causal  lifting form.    We shall see that the ability to factor off diagonal delay matrices  $\bsy{\Lambda}_n(z)$ is a major advantage the CCA holds over the causal version of the EEA lifting factorization method.

The ancient Euclidean Algorithm was a clever  idea for recursively reducing  the gcd of  ``large'' arguments to the gcd of  ``smaller'' arguments.  The notion of \emph{degree-reducing} solutions to LDEs over polynomial rings (Definition~\ref{defn:CausalComplement}) captures the size-reducing aspect of the EEA   without  its misleading focus on gcds.     Moreover,  the degree-reducing notion  applies to factorizations of arbitrary FIR PR filter banks whereas the  \emph{polyphase order-increasing} property employed in~\cite{Bris:10:GLS-I,Bris:10b:GLS-II} was found to be useful only for linear phase filter banks.
A degree-{reducing} decomposition  corresponds to a degree-{increasing}  synthesis, so we use the  neutral  term \emph{degree-lifting} to encompass both decomposition and synthesis.
The author holds that  this {degree-lifting} character of both the EEA and the CCA is what distinguishes \emph{lifting} factorizations within the much bigger universe of elementary matrix decompositions, a distinction not made in~\cite{DaubSwel98}.

\subsection{Overview of the Paper}\label{sec:Intro:Overview}
Section~\ref{sec:Study} introduces LGT(5,3), the  LeGall-Tabatabai 5-tap/3-tap piecewise-linear spline wavelet filter bank~\cite{LeGallTabatabai:88:Subband-coding-digital,ISO_15444_1}, and uses it to illustrate a connection between causality and uniqueness of lifting factorizations.
It is shown that the \hyperref[sec:Study:CCA]{Causal Complementation Algorithm} (CCA)  generates all  of the causal lifting factorizations formed by the  \hyperref[sec:Study:EEA]{causal  version of the EEA}  plus \hyperref[sec:Study:Other]{other factorizations not generated by it},  generalizing the degree-lifting aspects of the EEA.

To understand the connection between causality and uniqueness of  factorizations, Section~\ref{sec:LDE}  defines  linear Diophantine equations (LDEs) and  reviews the   ring theory needed to characterize the solution sets of \hyperref[thm:Homog_LDE]{homogeneous LDEs} and  inhomogeneous LDEs (the \hyperref[cor:AbsLiftingTheorem]{Abstract Lifting Theorem}).    Necessary and sufficient conditions are given for existence (but not uniqueness) of solutions to inhomogeneous LDEs in principal ideal domains (the \hyperref[thm:Inhomog_LDE]{Abstract Bezout Theorem}).  Section~\ref{sec:LDE:Euclidean} reviews  the differences between  polynomials and Laurent polynomials. Readers not interested in these algebraic underpinnings should skip ahead to Theorem~\ref{thm:LDDRT}.

\hyperref[defn:CausalComplement]{Degree-reducing causal complements}  are generalized in Definition~\ref{defn:SizeReducing} to include   LDEs over the Laurent polynomials.  By Lemma~\ref{lem:LDE_uniqueness} \emph{degree-reducing} solutions to \emph{polynomial} LDEs are unique thanks to a  \hyperref[deg_add_bound]{max-additive inequality} that \emph{fails}  for the Laurent order, explaining why  filters can have multiple Laurent-order-reducing complements.   Existence and uniqueness of degree-reducing solutions (in either unknown) to  polynomial LDEs   are  given by the \hyperref[thm:LDDRT]{Linear Diophantine Degree-Reduction Theorem} (LDDRT, Theorem~\ref{thm:LDDRT}), along with necessary and sufficient conditions  for both solutions to agree. Although quite elementary, the LDDRT  seems to have escaped attention in the literature, perhaps for lack of a motive for finding degree-reducing solutions.

Section~\ref{sec:LDE:Division} uses the LDDRT to prove  existence and uniqueness of quotients whose remainders satisfy degree-reducing inequalities and  user-specified divisibility constraints   (the \hyperref[cor:GDT]{Generalized Polynomial Division Theorem}, which also appears to have escaped previous notice).  This is specialized to the case of  remainders divisible by monomials and given a constructive proof in the  \hyperref[thm:SGDT]{Slightly Generalized Division Theorem},  yielding a  \hyperref[alg:SGDA]{Slightly Generalized Division Algorithm} (SGDA).

Section~\ref{sec:CubicBSpline} returns to the  filter bank setting and presents a case study using  CDF(7,5), a 7-tap/5-tap cubic B-spline wavelet filter bank~\cite{DaubSwel98}.
It is shown that the \hyperref[sec:CubicBSpline:EEA_Col1]{causal EEA} and   \hyperref[sec:CubicBSpline:CCA_Col1]{CCA factorizations} in column~1 using classical polynomial division are the same.  To produce a \emph{causal} analogue of a unimodular linear phase lifting factorization obtained by Daubechies and Sweldens~\cite[\S7.8]{DaubSwel98},
we  use the SGDA to factor off a diagonal delay matrix with the first CCA lifting step. \hyperref[sec:CubicBSpline:SGDA]{This factorization} is \emph{not} produced by running the causal EEA  in {any} row or column of  CDF(7,5).

Section~\ref{sec:Conclusions} summarizes the paper's contributions. 

%
\section{Case Study: The LeGall-Tabatabai Filter Bank}\label{sec:Study}
We begin by contrasting the  EEA and CCA approaches to factoring LGT(5,3), the 5-tap/3-tap LeGall-Tabatabai biorthogonal linear phase filter bank~\cite{LeGallTabatabai:88:Subband-coding-digital},
whose analog synthesis scaling function and mother wavelet generate  piecewise-linear B-splines.
JPEG~2000 Part~1 \cite{ISO_15444_1} specifies LGT(5,3)  via a lifting factorization of the
unimodular \emph{polyphase-with-advance} (PWA)  matrix representation $\mathbf{A}(z)$~\cite{DaubSwel98,BrisWohl06} for
its noncausal  whole-sample symmetric analysis filters. This lifting factorization is its unique WS lifting factorization in the  WS group of unimodular whole-sample symmetric transfer matrices~\cite[Definition~8]{Bris:10:GLS-I}. The factorization halves the number of  multiplications per unit input when implementing the filter bank, a big motivation for using lifting decompositions.
\begin{equation}\begin{split}
\label{LGT53_noncausal_filters}
A_0(z) &= -z^2/8 + z/4 + 3/4 + z^{-1}/4 - z^{-2}/8\\
A_1(z) &= -z^2/2 + z -1/2
\end{split}\end{equation}
\begin{align}
\mathbf{A}(z) 
&\eqdef
\begin{bmatrix}
\sst (-z + 6 - z^{-1})/8 & \sst \;(1 + z^{-1})/4 \\
\sst -(z + 1)/2	& \sst 1  \vspace{-1pt}	
\end{bmatrix}
=
\begin{bmatrix}
\sst 1 &\sst (1 + z^{-1})/4\\
\sst 0 &\sst 1 \vspace{-1pt}	
\end{bmatrix}
\hspace{-4pt}
\begin{bmatrix}
\sst 1 &\sst 0\\
\sst -(z+1)/2 &\sst 1 \vspace{-1pt}	
\end{bmatrix}\label{LGT_GLF}
\end{align}

We will work instead with the causally normalized filters
\begin{equation}\begin{split}
\label{LGT53_causal_filters}
H_0(z) &= (-1 + 2z^{-1} + 6z^{-2} + 2z^{-3} - z^{-4})/8,\\
H_1(z) &= (-1 + 2z^{-1} -z^{-2})/2
\end{split}\end{equation}
and their causal \emph{polyphase-with-delay} (PWD) analysis matrix  \cite{Vaid93,StrNgu96}, 
\begin{align}\label{LGT}
\mathbf{H}(z) 
&=
\begin{bmatrix}
 (-1 + 6z^{-1} - z^{-2})/8 &  \;(1 + z^{-1})/4 \\
 -(1 + z^{-1})/2	&  1 \vspace{-1pt}	
\end{bmatrix} ,
\quad  |\mathbf{H}(z)|  = z^{-1}.
\end{align}

Laurent polynomials can have multiple reduced-Laurent-order unimodular complements (shorter filters with which they form a unimodular filter bank). For instance, $A_0(z)$ also has the reduced-order \emph{nonlinear phase} unimodular complement  $A'_1(z) = -7/2 - z^{-1} + z^{-2}/2$, which happens to be causal.  $A_0(z)$ and $A'_1(z)$ have noncausal unimodular PWA matrix $\mathbf{A}'(z)$, 
\begin{align*}
\mathbf{A}'(z) 
&=
\begin{bmatrix}
 (-z + 6 - z^{-1})/8 &  \;(1 + z^{-1})/4 \\
 (-7 + z^{-1})/2	&  -z^{-1}  \vspace{-1pt}	
\end{bmatrix},
\end{align*}
while $H_0(z)$ and $H'_1(z)\eqdef A'_1(z)$ have causal PWD counterpart $\mathbf{H}'(z)$,
\begin{align*}
\mathbf{H}'(z) 
&=
\begin{bmatrix}
 (-1 + 6z^{-1} - z^{-2})/8 &  \;(1 + z^{-1})/4 \\
 (-7 + z^{-1})/2	&  -1 \vspace{-1pt}	
\end{bmatrix} .
\end{align*}
The  determinants of  $\mathbf{A}(z)$ and $\mathbf{A}'(z)$ are both~1, but   $|\mathbf{H}(z)|  = z^{-1}$ and $|\mathbf{H}'(z)| = 1$ distinguish between the two reduced-degree \emph{causal} complements of  $H_0(z)$.   
Theorem~\ref{thm:LDDRT} will imply that  $H_1(z)$ and $H'_1(z)$ are the \emph{unique} reduced-degree causal complements to $H_0(z)$  for these  determinants.  
This shows that the unimodular normalization employed by Daubechies and Sweldens~\cite{DaubSwel98} is discarding useful information about  filter banks.
Using Theorem~\ref{thm:LDDRT} and the connection between lifting factorization and degree-lifting causal complementation,  the  Diophantine perspective uncovers a connection between \emph{causality} and \emph{uniqueness} of lifting factorizations that allows the CCA to enumerate and generate  \emph{all} causal degree-lifting factorizations of a causal filter bank.

\subsection{Causal Factorization via the Causal Extended Euclidean Algorithm}\label{sec:Study:EEA}
There are four possible  factorizations based on running the causal EEA in either row or  column of $\mathbf{H}(z)$.  Our notation for the EEA is a compromise between several sources, including~\cite{DaubSwel98}, \cite{BachShallit:96:Algo-Number-Theory}, \cite{Shoup:05:Comp-Number-Theory}, and~\cite{GathenGerhard:13:Modern-Computer-Algebra}.

\subsubsection{EEA in Column~0}\label{sec:Study:EEA:Col0}
Initialize  remainders
$r_0 \eqdef H_{00}(z) = (-1+6z^{-1}-z^{-2})/8$ and $r_1 \eqdef H_{10}(z) = -(1+z^{-1})/2$.
Iterate using the  polynomial division algorithm, $r_0 = q_0 r_1 + r_2$, 
\begin{align}\label{r2}
q_0 = (-7+z^{-1})/4,\;  r_2 = -1,\text{ and } \deg(r_2)=0<\deg(r_1)=1.
\end{align}
Define the matrix
\begin{align}\label{M0}
\mathbf{M}_0 &\eqdef
\begin{bmatrix}
q_0 & 1\\
1 & 0
\end{bmatrix},\text{ so that }
\begin{pmatrix}
r_0\\
r_1
\end{pmatrix}
=
\mathbf{M}_0
\begin{pmatrix}
r_1\\
r_2
\end{pmatrix}.
\end{align}
Remainder $r_2$ is invertible so the next division step yields
\begin{align}
r_3 = r_1 - q_1 r_2 = 0\mbox{ with }q_1 = r_1 / r_2 = (1 + z^{-1})/2,\label{r3}
\end{align}
where $\deg(r_3)=\deg(0)\eqdef -\infty < \deg(r_2)$. 
Define
\begin{align}
\mathbf{M}_1 \eqdef
\begin{bmatrix}
q_1 & 1\\
1 & 0
\end{bmatrix},\text{ so that }
\begin{pmatrix}
r_1\\
r_2
\end{pmatrix}
=
\mathbf{M}_1
\begin{pmatrix}
r_2\\
r_3
\end{pmatrix}.\label{M1}
\end{align}

Iteration terminates since $r_3=0$, and \eqref{M0} and~\eqref{M1} imply
\begin{align}\label{M0M1}
\begin{pmatrix}
r_0\\
r_1
\end{pmatrix}
=
\mathbf{M}_0\mathbf{M}_1
\begin{pmatrix}
r_2\\
0
\end{pmatrix}.
\end{align}
Put two swap matrices $\mathbf{J}$~\eqref{swap_matrix} between $\mathbf{M}_0$ and $\mathbf{M}_1$ to get lifting matrices $\mathbf{S}_0$ and $\mathbf{S}_1$,
\begin{align*}
\mathbf{M}_0\mathbf{M}_1
&=
(\mathbf{M}_0\mathbf{J})(\mathbf{J}\mathbf{M}_1)
=
\begin{bmatrix}
1 & q_0\\
0 & 1
\end{bmatrix}
\hspace{-3pt}
\begin{bmatrix}
1 & 0\\
q_1 & 1
\end{bmatrix}
=
\mathbf{S}_0\mathbf{S}_1.
\end{align*}
As in the unimodular approach~\cite{DaubSwel98}, form an \emph{augmentation matrix} $\mathbf{H}'(z)$ by augmenting~\eqref{M0M1} with  causal filters $a_0$ and $a_1$ determined by the condition
\begin{align}
\mathbf{H}'(z)&\eqdef
\begin{bmatrix}
r_0 & a_0\\
r_1 & a_1\vspace{1pt}  
\end{bmatrix}
=
\mathbf{S}_0\mathbf{S}_1
\begin{bmatrix}
r_2 & 0\\
0 & |\mathbf{H}(z)|/r_2
\end{bmatrix}\label{S1S2}\nonumber\\
&=
\begin{bmatrix}
1 & (-7+z^{-1})/4\\
0 & 1 \vspace{-1pt}	
\end{bmatrix}
\hspace{-3pt}
\begin{bmatrix}
1 & 0\\
(1 + z^{-1})/2 & 1 \vspace{-1pt}	
\end{bmatrix}
\hspace{-3pt}
\begin{bmatrix}
-1 & 0\\
0 & -z^{-1} \vspace{-1pt}	
\end{bmatrix}\\
&=
\begin{bmatrix}
 (-1 + 6z^{-1} - z^{-2})/8 &  \;(7z^{-1} - z^{-2})/4 \\
 -(1 + z^{-1})/2	&  -z^{-1} \vspace{-1pt}	
\end{bmatrix} .\nonumber
\end{align}
$|\mathbf{H}'(z)|=|\mathbf{H}(z)|$ by~\eqref{S1S2} and the matrices agree in column~0, so the  Lifting Theorem~\cite{Herley93,VetterliHerley:92:Wavelets-filter-banks,Sweldens96} (Corollary~\ref{cor:AbsLiftingTheorem} below) says that $\mathbf{H}(z)$ can be lifted from $\mathbf{H}'(z)$ by a causal  lifting update to column~1, $\mathbf{H}(z)=\mathbf{H'}(z)\mathbf{S}(z)$, 
%
\begin{align}\label{S1S2LiftingThm}
\mathbf{H}(z)
&=
\begin{bmatrix}
H_{00} & H_{01}\\
H_{10} & H_{11} \vspace{-1pt}
\end{bmatrix}
=
\begin{bmatrix}
r_0 & a_0\\
r_1 & a_1\vspace{1pt}  
\end{bmatrix}
\hspace{-3pt}
\begin{bmatrix}
1 & S\\
0 & 1 \vspace{-1pt}
\end{bmatrix}
\text{\ iff\ }
\left\{\begin{array}{l}
H_{01} = r_0S + a_0 \\
H_{11} = r_1S + a_1\, . 
\end{array}\right.
\end{align}
Compute $H_{01} - a_0 = (1 - 6z^{-1} + z^{-2})/4 = -2r_0$, so $S = -2$.
The resulting factorization in standard  causal lifting form~\eqref{std_causal_form} based on~\eqref{S1S2}--\eqref{S1S2LiftingThm} is
\begin{align}\label{LGT_EEA_col0}
\mathbf{H}(z) 
&= -
\begin{bmatrix}
 1 &  (-7 + z^{-1})/4\\
 0 &  1%
\end{bmatrix}
\hspace{-4pt}
\begin{bmatrix}
 1 &  0\\
 (1+z^{-1})/2 &  1%
\end{bmatrix}
\hspace{-4pt}
\begin{bmatrix}
 1 &  0\\
 0 &  z^{-1}%
\end{bmatrix}
\hspace{-4pt}
\begin{bmatrix}
 1 &  -2\\
 0 &  1%
\end{bmatrix} .
\end{align}

\subsubsection{EEA in Column~1}\label{sec:Study:EEA:Col1}
Initialize  
$r_0 \eqdef H_{01}(z) = (1+z^{-1})/4$ and $r_1 \eqdef H_{11}(z) = 1$.  Polynomial  division yields
$q_0 = r_0/r_1 = (1+z^{-1})/4$, and  $r_2=0$  so
\begin{align}
\begin{pmatrix}
r_0\\
r_1
\end{pmatrix}
=
\mathbf{M}_0
\begin{pmatrix}
r_1\\
0
\end{pmatrix} \text{ with }
\mathbf{M}_0 \eqdef
\begin{bmatrix}
q_0 & 1\\
1 & 0
\end{bmatrix},\; |\mathbf{M}_0|=-1.  \label{M0_col1}
\end{align}
Augment~\eqref{M0_col1} in column~0 with causal filters $a_0$ and $a_1$ defined by
\begin{align}\label{G_LGT_col1}
\mathbf{H}'(z)
&\eqdef
\begin{bmatrix}
a_0 & r_0\\
a_1 & r_1
\end{bmatrix}
=
\mathbf{M}_0
\begin{bmatrix}
0 & r_1\\
|\mathbf{H}|/r_1 & 0
\end{bmatrix}
=
\mathbf{M}_0\,\mathbf{J}^2
\begin{bmatrix}
0 & r_1\\
|\mathbf{H}|/r_1 & 0
\end{bmatrix}\\
&=
\begin{bmatrix}
1\; & (1+z^{-1})/4 \\
0 & 1 \vspace{-1pt}	
\end{bmatrix}
\hspace{-3pt}
\begin{bmatrix}
z^{-1} & 0\\
0 & 1
\end{bmatrix}.\nonumber
\end{align}
$|\mathbf{H}'(z)|=|\mathbf{H}(z)|$ and the matrices agree in column~1 so $\mathbf{H}(z)$ can be lifted from  $\mathbf{H}'(z)$ by a causal lifting update to column~0,
\begin{align}\label{G_LGT_col1_LiftingThm}
\mathbf{H}(z) 
&=
\begin{bmatrix}
H_{00} & H_{01}\\
H_{10} & H_{11} \vspace{-1pt}
\end{bmatrix}
=
\begin{bmatrix}
a_0 & r_0 \\
a_1 & r_1 \vspace{1pt}  
\end{bmatrix}
\hspace{-3pt}
\begin{bmatrix}
1 & 0\\
S & 1 \vspace{-1pt}
\end{bmatrix}
\text{\ iff\ }
\left\{\begin{array}{l}
H_{00} = r_0S + a_0 \\
H_{10} = r_1S + a_1 \,.
\end{array}\right.
\end{align}
$H_{00} - a_0 = -(1 + 2z^{-1} + z^{-2})/8 = r_0S$ for $S(z) = -(1+z^{-1})/2$ so by~\eqref{G_LGT_col1_LiftingThm}   
\begin{align}\label{LGT_EEA_col1}
\mathbf{H}(z) &=
\begin{bmatrix}
1 & (1+ z^{-1})/4\\
0 & 1\vspace{-1pt}
\end{bmatrix}
\hspace{-3pt}
\begin{bmatrix}
z^{-1} & 0\\
0 & 1\vspace{-1pt}
\end{bmatrix}
\hspace{-3pt}
\begin{bmatrix}
1 & 0\\
 -(1 + z^{-1})/2	&  1 \vspace{-1pt}
\end{bmatrix}.
\end{align}
\rem  This is a \emph{causal} analogue of the noncausal linear phase lifting~\eqref{LGT_GLF} for the unimodular LGT(5,3) analysis bank; its causal linear phase lifting filters differ from the corresponding lifting filters in~\eqref{LGT_GLF} by at most  delays.

\subsubsection{EEA in Row~0}\label{sec:Study:EEA:Row0}
Initialize  $r_0 \eqdef H_{00}(z)$ 
and 
$r_1 \eqdef H_{01}(z)$. 
Division yields $q_0 = (7 - z^{-1})/2$ and $r_2 = r_0 - r_1 q_0 = -1$,
\begin{align}
(r_0,\, r_1) = (r_1,\, r_2)\,\mathbf{M}_0, \text{ with }
\mathbf{M}_0 
\eqdef
\begin{bmatrix}
q_0 & 1\\
1 & 0
\end{bmatrix}.\label{M0_row0}
\end{align}
Since $r_2=-1$ is invertible, we get $q_1 = r_1 / r_2 = -(1 + z^{-1})/4$, $r_3=0$, and
\begin{align}
(r_1,\, r_2) = (r_2,\, 0)\, \mathbf{M}_1, \text{ with }
\mathbf{M}_1 \eqdef
\begin{bmatrix}
q_1 & 1\\
1 & 0
\end{bmatrix}.\label{M1_row0}
\end{align}
Combine~\eqref{M0_row0} and~\eqref{M1_row0} and define 
$\mathbf{S}_0\eqdef \mathbf{J}\mathbf{M}_0,\;\mathbf{S}_1\eqdef  \mathbf{M}_1\mathbf{J}$ to express the top row as 
$(r_0,\, r_1) = (r_2,\, 0)\,\mathbf{S}_1\mathbf{S}_0.$
Augment with  filters $a_0$ and $a_1$,
\begin{align*}
\mathbf{H}'(z) 
&\eqdef
\begin{bmatrix}
r_0 & r_1\\
a_0 & a_1
\end{bmatrix}
=
\begin{bmatrix}
r_2 & 0\\
0 & \;|\mathbf{H}|/r_2
\end{bmatrix}
\mathbf{S}_1\mathbf{S}_0
=
\begin{bmatrix}
\sst (-1 + 6z^{-1} - z^{-2})/8 &\sst   \;(1 + z^{-1})/4 \\
\sst  (-7z^{-1} + z^{-2})/2	&\sst  - z^{-1} \vspace{-1pt}	
\end{bmatrix}. 
\end{align*}
$\mathbf{H}$ can be lifted from augmentation matrix $\mathbf{H}'$ by updating row~1, $\mathbf{H}=\mathbf{S}\mathbf{H}'$,
\begin{align*}
\begin{bmatrix}
H_{00} & H_{01}\\
H_{10} & H_{11} \vspace{-1pt}
\end{bmatrix}
&=
\begin{bmatrix}
1 & 0\\
S & 1 \vspace{-1pt}
\end{bmatrix}
\hspace{-3pt}
\begin{bmatrix}
r_0 & r_1 \\
a_0 & a_1 \vspace{1pt}
\end{bmatrix}
\text{\ iff\ }
\left\{\begin{array}{l}
H_{10} = r_0S + a_0 \\
H_{11} = r_1S + a_1. \vspace{-1pt}
\end{array}\right.
\end{align*}
This implies $S(z) = 4$, and the causal lifting factorization is 
\begin{align}\label{LGT_EEA_row0}
\mathbf{H}(z) 
&= -
\begin{bmatrix}
 1 &  0\\
 4 &  1%
\end{bmatrix}
\hspace{-4pt}
\begin{bmatrix}
 1 &  0\\
 0 &  z^{-1}%
\end{bmatrix}
\hspace{-4pt}
\begin{bmatrix}
 1 &  -(1+z^{-1})/4\\
 0 &  1%
\end{bmatrix}
\hspace{-4pt}
\begin{bmatrix}
 1 &  0\\
 (7 - z^{-1})/2 &  1%
\end{bmatrix} .
\end{align}

\subsubsection{EEA in Row~1}\label{sec:Study:EEA:Row1}  
This option also yields  the causal linear phase lifting factorization~\eqref{LGT_EEA_col1}.

\subsection{Causal Factorization via the Causal Complementation Algorithm}\label{sec:Study:CCA}
 We begin by using the CCA to reproduce all of the EEA factorizations.

\subsubsection{CCA With Division in Column~0}\label{sec:Study:CCA:Col0} 
Initialize   quotient matrix $\mathbf{Q}_0(z)\eqdef \mathbf{H}(z)$ and  find an initial lifting   of the form
\begin{equation}\label{LGT_left_update}
\mathbf{Q}_0(z) =
\begin{bmatrix}
\sst (-1 + 6z^{-1} - z^{-2})/8 &\sst   \;(1 + z^{-1})/4 \\
\sst  -(1 + z^{-1})/2	&\sst   1 \vspace{-1pt}	
\end{bmatrix}
=
\begin{bmatrix}
\sst E_0 &\sst  E_1\\
\sst F_0 &\sst  F_1\vspace{-1pt}
\end{bmatrix}
=
\begin{bmatrix}
\sst 1 &\sst  S\\
\sst 0 &\sst  1\vspace{-1pt}%
\end{bmatrix}
\hspace{-4pt}
\begin{bmatrix}
\sst R_0 &\sst  R_1\\
\sst F_0 &\sst  F_1\vspace{-1pt}%
\end{bmatrix}.
\end{equation}
Set  $E_0\leftarrow H_{00}$, $F_0\leftarrow H_{10}$,   
and divide $F_0$ into $E_0$ to get $E_0 = F_0S + R_0$  with
\begin{equation}\label{LGT_col0_step0_div}
S(z) = (-7 + z^{-1})/4\text{ and } R_0(z) = -1,\text{ as  in the EEA~\eqref{r2}.}
\end{equation}
Set $R_1 \leftarrow E_1 - F_1S = 2$  to get a row reduction,  $(R_0,R_1) = (E_0,E_1) - S(F_0,F_1)$.

$R_0$ and $R_1$ are coprime,  so the first lifting step~\eqref{LGT_left_update} is
\begin{align}\label{LGT_matrices_col0_step0}
\mathbf{Q}_0(z) 
&=
\begin{bmatrix}
1 & (-7 + z^{-1})/4\\
0 & 1\vspace{-1pt}%
\end{bmatrix}
\hspace{-3pt}
\begin{bmatrix}
-1 & 2\\
-(1+z^{-1})/2 & \,1\vspace{-1pt}%
\end{bmatrix} 
=
\mathbf{V}_0(z)\mathbf{Q}_1(z).
\end{align}
Reset  $E_j\leftarrow F_j$ and $F_j\leftarrow R_j$ in $\mathbf{Q}_1(z)$ and divide again in column~0 to get 
\begin{align}\label{LGT_col0_step1_form}
\mathbf{Q}_1(z) 
=
\begin{bmatrix}
-1 & 2\\
-(1+z^{-1})/2 & 1\vspace{-1pt}%
\end{bmatrix}
=
\begin{bmatrix}
F_0 & F_1\\
E_0 & E_1\vspace{-1pt}
\end{bmatrix}
=
\begin{bmatrix}
1 & 0\\
S & 1\vspace{-2pt}%
\end{bmatrix}
\hspace{-3pt}
\begin{bmatrix}
F_0 & F_1\\
R_0 & R_1\vspace{-1pt}%
\end{bmatrix}.
\end{align}
Dividing  $F_0=-1$ into  $E_0=-(1+z^{-1})/2$ as in~\eqref{r3}  gives $E_0 = F_0S + R_0$, 
\begin{equation}\label{LGT_col0_step1_div}
S(z) = (1+z^{-1})/2\text{ and } R_0 =0,\; \deg(R_0) \eqdef -\infty < \deg(F_0). 
\end{equation}
Set  $R_1 \leftarrow E_1 - F_1S = -z^{-1}$; $\gcd(R_0,R_1)=z^{-1}$ so factor it out of $\mathbf{Q}_2(z)$,
\begin{align}
\mathbf{Q}_1(z) 
&=
\begin{bmatrix}
\sst 1 &\sst  0\\
\sst (1+z^{-1})/2 &\sst  \,1\vspace{-1pt}%
\end{bmatrix}
\hspace{-4pt}
\begin{bmatrix}
\sst -1 &\sst  2\\
\sst 0 &\sst  -z^{-1}\vspace{-1pt}
\end{bmatrix}
=
\begin{bmatrix}
\sst 1 &\sst  0\\
\sst (1+z^{-1})/2 &\sst  \,1\vspace{-1pt}%
\end{bmatrix}
\hspace{-4pt}
\begin{bmatrix}
\sst 1 &\sst  0\\
\sst 0 &\sst  z^{-1}\vspace{-1pt}
\end{bmatrix}
\hspace{-4pt}
\begin{bmatrix}
\sst -1 &\sst  2\\
\sst 0 &\sst  -1\vspace{-1pt}
\end{bmatrix}\label{LGT_col0_step2} \\
&=
\mathbf{V}_1(z)\bsy{\Delta}_1(z)\mathbf{Q}_2(z),\text{ where }\bsy{\Delta}_1(z)=\diag(1, z^{-1}). \nonumber
\end{align}
Combining~\eqref{LGT_matrices_col0_step0} and~\eqref{LGT_col0_step2} and dividing $\mathbf{Q}_2(z)$ by $-1$ to get a proper lifting step yields the same factorization~\eqref{LGT_EEA_col0} obtained using the EEA in column~0,
%
\begin{align}\label{LGT_col0_factorization} 
\mathbf{H}(z) 
&= -
\begin{bmatrix}
 1 &  (-7 + z^{-1})/4\\
 0 &  1%
\end{bmatrix}
\hspace{-4pt}
\begin{bmatrix}
 1 &  0\\
 (1+z^{-1})/2 &  1%
\end{bmatrix}
\hspace{-4pt}
\begin{bmatrix}
 1 &  0\\
 0 &  z^{-1}%
\end{bmatrix}
\hspace{-4pt}
\begin{bmatrix}
 1 &  -2\\
 0 &  1%
\end{bmatrix} .
\end{align}

\subsubsection{CCA With Division in Column~1}\label{sec:Study:CCA:Col1} 
Initialize  $E_1\leftarrow H_{01}$ and $F_1\leftarrow H_{11}$. Divide $F_1$ into $E_1$ in~\eqref{LGT_left_update} to get
$E_1 = F_1S + R_1$ where $S =  (1 + z^{-1})/4$ and $R_1 = 0.$  Set $R_0 \leftarrow E_0 - F_0S = z^{-1}$;  factoring out $\gcd(R_0,R_1)=z^{-1}$ yields the  lifting~\eqref{LGT_EEA_col1} obtained by the EEA,
\begin{align}\label{LGT_col1_factorization}
\mathbf{H}(z) =
\begin{bmatrix}
1 & (1 + z^{-1})/4\\
0 & 1\vspace{-1pt}%
\end{bmatrix}
\hspace{-3pt}
\begin{bmatrix}
z^{-1} & 0\\
0 & 1\vspace{-1pt}%
\end{bmatrix}
\hspace{-3pt}
\begin{bmatrix}
1 & 0\\
-(1+z^{-1})/2 & \,1\vspace{-1pt}%
\end{bmatrix} .
\end{align}

\subsubsection{CCA With Division in Row~0}\label{sec:Study:CCA:Row0} 
Initialize  $E_0\leftarrow H_{00}$  and $F_0\leftarrow H_{01}$ and divide $F_0$ into $E_0$ to get
\begin{align}
&\mathbf{Q}_0(z) =
\begin{bmatrix}
E_0 & F_0\\
E_1 & F_1\vspace{-1pt}
\end{bmatrix}
=
\begin{bmatrix}
R_0 & F_0\\ 
R_1 & F_1\vspace{-1pt}%
\end{bmatrix}
\hspace{-3pt}
\begin{bmatrix}
1 & 0\\
S & 1\vspace{-1pt}%
\end{bmatrix},
\;S = (7 - z^{-1})/2,\;R_0 = -1.\label{row0_step0}
\end{align}
Set $R_1 \leftarrow E_1 - F_1S = -4$.  The first factorization step is 
\begin{align}\label{LGT_matrices_row0_step0}
\mathbf{Q}_0(z) 
&= \mathbf{Q}_1(z)\mathbf{V}_0(z)
=
\begin{bmatrix}
-1 & (1+z^{-1})/4\\
-4 & \,1\vspace{-1pt}%
\end{bmatrix} 
\hspace{-3pt}
\begin{bmatrix}
1 & 0\\
(7 - z^{-1})/2 & \,1\vspace{-1pt}%
\end{bmatrix}.
\end{align}
Reset the labels $E_j\leftarrow F_j$ and $F_j\leftarrow R_j$ in $\mathbf{Q}_1(z)$, so that
\begin{align}\label{LGT_matrices_row0_step1}
\mathbf{Q}_1(z) =
\begin{bmatrix}
-1 & (1+z^{-1})/4\\
-4 & \,1\vspace{-1pt}%
\end{bmatrix} 
=
\begin{bmatrix}
F_0 & E_0\\
F_1 & E_1\vspace{-1pt}
\end{bmatrix}.
\end{align}

Divide  $F_0=-1$ into  $E_0=(1+z^{-1})/4$,
\begin{align}
E_0 &= F_0S + R_0,\; S = -(1+z^{-1})/4,\text{ and }R_0=0.\label{row0_step1}
\end{align}
The first two division steps~\eqref{row0_step0} and \eqref{row0_step1} are identical to the first two steps in the row~0 EEA calculation, \eqref{M0_row0} and \eqref{M1_row0}.   Set  $R_1 \leftarrow E_1 - F_1S = -z^{-1}$,  
\begin{align*}
\mathbf{Q}_1(z) 
=
\begin{bmatrix}
F_0 & R_0\\
F_1 & R_1\vspace{-1pt}
\end{bmatrix}
\hspace{-3pt}
\begin{bmatrix}
1 & S\\
0 & 1\vspace{-1pt}%
\end{bmatrix}
=
\begin{bmatrix}
-1 & 0\\
-4 & -z^{-1}\vspace{-1pt}%
\end{bmatrix} 
\hspace{-3pt}
\begin{bmatrix}
1 & \,-(1+z^{-1})/4\\
0 & 1\vspace{-1pt}%
\end{bmatrix}.
\end{align*}
Factor out $\gcd(R_0,R_1)=z^{-1}$; the resulting    factorization agrees with~\eqref{LGT_EEA_row0},
\begin{align}\label{LGT_row0_factorization}
\mathbf{H}(z) 
&= -
\begin{bmatrix}
 1 &  0\\
 4 &  1
\end{bmatrix} 
\hspace{-4pt}
\begin{bmatrix}
 1 &  0\\
 0 &  z^{-1}
\end{bmatrix} 
\hspace{-4pt}
\begin{bmatrix}
 1 &  -(1+z^{-1})/4\\
 0 &  1
\end{bmatrix}
\hspace{-4pt}
\begin{bmatrix}
 1 &  0\\
 (7 - z^{-1})/2 &  1
\end{bmatrix}.
\end{align}

\subsubsection{CCA With Division in Row~1}\label{sec:Study:CCA:Row1} 
 The  factorization is identical to the linear phase factorization~\eqref{LGT_EEA_col1},  \eqref{LGT_col1_factorization}.

\subsection{Other Degree-Lifting Factorizations via the CCA}\label{sec:Study:Other} 
It is convenient to have two  algebraic tools for manipulating lifting cascades.
\begin{definition}[cf.~\cite{Bris:10:GLS-I}, eq.~(27)]
\label{defn:DIF}
Let  $\mathbf{D}_{\kappa_0,\kappa_1}\eqdef\diag(\kappa_0,\kappa_1)$; $\kappa_0,\kappa_1\neq 0$. 
The  \emph{diagonal inner automorphism} 
$\gamma_{\kappa_0,\kappa_1}$   is the matrix automorphism $\gamma_{\kappa_0,\kappa_1}\mathbf{A} \eqdef \mathbf{D}_{\kappa_0,\kappa_1}\mathbf{A}\,\mathbf{D}^{-1}_{\kappa_0,\kappa_1}$, which satisfies
\begin{align}
\gamma_{\kappa_0,\kappa_1}\!
\begin{bmatrix}
a & b \\
c & d\vspace{-2pt} 
\end{bmatrix} 
&\eqdef  
\begin{bmatrix}
\kappa_0 & 0\\
0 & \kappa_1%
\end{bmatrix}
\hspace{-3pt}
\begin{bmatrix}
a & b \\
c & d\vspace{-2pt} 
\end{bmatrix} 
\hspace{-3pt}
\begin{bmatrix}
\kappa_0^{-1\rule{0pt}{6pt}} & 0\\
0 & \kappa_1^{-1}%
\end{bmatrix}
=
\begin{bmatrix}
a &  \kappa_0 \kappa_1^{-1\rule{0pt}{6pt}} b \\
\kappa_0^{-1} \kappa_1 c &  d%
\end{bmatrix} ,\label{DIF_automorphism}\\
\mathbf{D}_{\kappa_0,\kappa_1}\mathbf{A} 
&= (\gamma_{\kappa_0,\kappa_1}\mathbf{A})\mathbf{D}_{\kappa_0,\kappa_1}\mbox{ and } 
\mathbf{A}\,\mathbf{D}_{\kappa_0,\kappa_1} 
= \mathbf{D}_{\kappa_0,\kappa_1}\,\gamma^{-1}_{\kappa_0,\kappa_1}\mathbf{A}.\label{symbolic_DIF}
\end{align}
\end{definition}

\begin{definition}\label{defn:dbl_transpose}
The \emph{double transpose automorphism}   is    $\mathbf{A}^{\!\dbldag} \eqdef \mathbf{J}\mathbf{A}\mathbf{J}$,  which satisfies
\begin{align}
\begin{bmatrix}
a & b \\
c & d\vspace{-2pt} 
\end{bmatrix}^{\dbldag}
&\eqdef  
\begin{bmatrix}
0 & 1\\
1 & 0\vspace{-2pt} 
\end{bmatrix}
\hspace{-2pt}
\begin{bmatrix}
a & b \\
c & d\vspace{-2pt} 
\end{bmatrix}
\hspace{-2pt}
\begin{bmatrix}
0 & 1\\
1 & 0\vspace{-2pt} 
\end{bmatrix}
=
\begin{bmatrix}
d & c \\
b & a\vspace{-2pt} 
\end{bmatrix},\label{dbl_transpose_automorphism}\\
\mathbf{J}\mathbf{A} &= \mathbf{A}^{\!\dbldag}\mathbf{J}\mbox{ and } 
\mathbf{A}\mathbf{J} = \mathbf{J}\mathbf{A}^{\!\dbldag}.\label{dbl_transpose}
\end{align}
\end{definition}

\subsubsection{Switching Between Rows}\label{sec:Study:Other:Rows}
We now construct CCA factorizations that are {different} from those obtained using the EEA by exploiting  options that have no obvious EEA analogues.  
Consider~\eqref{LGT_matrices_row0_step0},  
\mbox{$\mathbf{Q}_0(z) \eqdef \mathbf{H}(z) = \mathbf{Q}_1(z)\mathbf{V}_0(z)$},
obtained by  dividing in row~0.  
In Section~\ref{sec:Study:CCA:Row0} we  factored $\mathbf{Q}_1(z)$ by dividing $F_0=-1$ into  $E_0=(1+z^{-1})/4$ in~\eqref{LGT_matrices_row0_step1}.  We are not obliged to continue dividing in row~0, however. Since
\[  \deg|\mathbf{Q}_1(z)| = 1 > \deg(F_0) + \deg(F_1) = 0,  \]  
Theorem~\ref{thm:LDDRT} (below) implies that  division in row~1 will produce a \emph{different} lifting step than~\eqref{row0_step1}. Therefore, divide $F_1$ into $E_1$ in $\mathbf{Q}_1(z)$ to get $E_1 = F_1S + R_1$, where $S = -1/4$ and $R_1 = 0$.
Set $R_0 \leftarrow E_0 - F_0S = z^{-1}/4$ and factor $\gcd(R_0,R_1)=z^{-1}$ out of column~1 of the quotient matrix,
\begin{align*}
\mathbf{Q}_1(z) 
&=
\begin{bmatrix}
F_0 & R_0\\
F_1 & R_1\vspace{-1pt}
\end{bmatrix}
\hspace{-3pt}
\begin{bmatrix}
1 & S\\
0 & 1\vspace{-1pt}%
\end{bmatrix}
=
\begin{bmatrix}
-1 & z^{-1}/4\\
-4 & 0\vspace{-1pt}%
\end{bmatrix} 
\hspace{-3pt}
\begin{bmatrix}
1 & \,-1/4\\
0 & 1\vspace{-1pt}%
\end{bmatrix}
=
\begin{bmatrix}
-1 & 1/4\\
-4 & 0\vspace{-1pt}%
\end{bmatrix} 
\hspace{-3pt}
\begin{bmatrix}
1 & 0\\
0 & z^{-1}\vspace{-1pt}%
\end{bmatrix}
\hspace{-3pt}
\begin{bmatrix}
1 & \,-1/4\\
0 & 1\vspace{-1pt}%
\end{bmatrix}\\
&=
\mathbf{Q}_2(z)\bsy{\Delta}_1(z)\mathbf{V}_1(z),\text{ where }
\bsy{\Delta}_1(z) = \diag(1, z^{-1}).
\end{align*}
$\mathbf{Q}_2(z)$  factors into a diagonal gain matrix, a  lifting step, and a swap matrix,
\begin{align*}
\mathbf{Q}_2(z) 
&=
\begin{bmatrix}
-1 & 1/4\\
-4 & 0\vspace{-1pt}%
\end{bmatrix} 
=
\begin{bmatrix}
1/4 & 0\\
0 & -4\vspace{-1pt}%
\end{bmatrix}
\hspace{-3pt}
\begin{bmatrix}
1 & -4\\
0 & 1\vspace{-1pt}%
\end{bmatrix}
\hspace{-3pt}
\begin{bmatrix}
0 & 1\\
1 & 0
\end{bmatrix} 
=
\diag(1/4,-4)\mathbf{V}_2(z)\,\mathbf{J} .
\end{align*}
Include the other matrices to get a causal lifting factorization~\eqref{LGT_switch_between_rows} that  is different from those obtained using the EEA in Section~\ref{sec:Study:EEA}.
\begin{align}\label{LGT_switch_between_rows} 
\mathbf{H}(z) 
&= 
\diag(1/4,-4)\mathbf{V}_2(z)\,\mathbf{J}\,\bsy{\Delta}_1(z)\mathbf{V}_1(z)\mathbf{V}_0(z)\nonumber\\
&=
\diag(1/4,-4)\mathbf{V}_2(z)\bsy{\Delta}_1^{\!\dbldag}(z)\mathbf{V}_1^{\!\dbldag}(z)\mathbf{V}_0^{\!\dbldag}(z)\,\mathbf{J}\text{\ \ by~\eqref{dbl_transpose}}\nonumber\\
&=
\begin{bmatrix}
\sst 1/4 & \sst 0\\
\sst 0 & \sst -4
\end{bmatrix}
\hspace{-4pt}
\begin{bmatrix}
\sst 1 & \sst -4\\
\sst 0 & \sst 1
\end{bmatrix}
\hspace{-4pt}
\begin{bmatrix}
\sst z^{-1} & \sst 0\\
\sst 0 & \sst 1
\end{bmatrix}
\hspace{-4pt}
\begin{bmatrix}
\sst 1 & \sst 0\\
\sst -1/4 & \sst 1
\end{bmatrix}
\hspace{-4pt}
\begin{bmatrix}
\sst 1 & \sst (7 - z^{-1})/2\\
\sst 0 & \sst 1
\end{bmatrix}
\hspace{-4pt}
\begin{bmatrix}
\sst 0 & \sst 1\\
\sst 1 & \sst 0
\end{bmatrix}. 
\end{align}

\subsubsection{Switching Between Columns}\label{sec:Study:Other:Cols}
Consider~\eqref{LGT_matrices_col0_step0},
\mbox{$\mathbf{Q}_0(z) = \mathbf{V}_0(z)\mathbf{Q}_1(z)$,} 
the initial lifting step obtained by dividing in column~0.  
In~\eqref{LGT_col0_step1_form} we see that the first quotient $\mathbf{Q}_1(z)$ satisfies $\deg|\mathbf{Q}_1(z)| = 1 > \deg(F_0) + \deg(F_1) = 0$ 
so Theorem~\ref{thm:LDDRT} implies that  division  in column~1 of $\mathbf{Q}_1(z)$ yields a \emph{different} lifting step than~\eqref{LGT_col0_step1_div}.  Division in column~1 yields $S=E_1/F_1=1/2$, $R_1=0$,  and  $R_0 \leftarrow E_0 - F_0S = -z^{-1}/2$ so
\begin{align*}
\mathbf{Q}_1(z) 
&=
\begin{bmatrix}
-1 & 2\\
-(1+z^{-1})/2 & 1\vspace{-1pt}%
\end{bmatrix}
=
\begin{bmatrix}
1 & 0\\
S & 1\vspace{-1pt}
\end{bmatrix}
\hspace{-3pt}
\begin{bmatrix}
F_0 & F_1\\
R_0 & R_1\vspace{-1pt}
\end{bmatrix}
=
\begin{bmatrix}
1 & 0\\
1/2 & 1\vspace{-1pt}%
\end{bmatrix}
\hspace{-3pt}
\begin{bmatrix}
1 & 0\\
0 & z^{-1}\vspace{-1pt}%
\end{bmatrix}
\hspace{-3pt}
\begin{bmatrix}
-1 & 2\\
-1/2 & 0\vspace{-1pt}%
\end{bmatrix} \\
&=
\mathbf{V}_1(z)\bsy{\Delta}_1(z)\mathbf{Q}_2(z).
\end{align*}
$\mathbf{Q}_2(z)$ can be factored into a gain matrix, a  lifting step, and a swap matrix,
\begin{align*}
\mathbf{Q}_2(z) 
=
\begin{bmatrix}
-1 & 2\\
-1/2 & 0\vspace{-1pt}%
\end{bmatrix} 
=
\begin{bmatrix}
2 & 0\\
0 & -1/2\vspace{-1pt}%
\end{bmatrix}
\hspace{-3pt}
\begin{bmatrix}
1 & -1/2\\
0 & 1\vspace{-1pt}%
\end{bmatrix}
\hspace{-3pt}
\begin{bmatrix}
0 & 1\\
1 & 0\vspace{-1pt}%
\end{bmatrix}
=
\diag(2,-1/2)\mathbf{V}_2(z)\,\mathbf{J} .
\end{align*}
The complete lifting factorization for $\mathbf{H}(z)$, which  differs from~\eqref{LGT_switch_between_rows} and from  the EEA factorizations in Section~\ref{sec:Study:EEA}, is
\begin{align*}
\mathbf{H}(z)
&= 
\mathbf{V}_0(z)\mathbf{V}_1(z)\bsy{\Delta}_1(z)\,\diag(2,-1/2)\mathbf{V}_2(z)\,\mathbf{J}\nonumber\\
&= 
\diag(2,-1/2)\bigl(\gamma^{-1}_{\ssst 2,-1/2}\mathbf{V}_0(z)\bigr)\bigl(\gamma^{-1}_{\ssst 2,-1/2}\mathbf{V}_1(z)\bigr) \bsy{\Delta}_1(z)\,\mathbf{V}_2(z)\,\mathbf{J}\text{ by~\eqref{symbolic_DIF},}\nonumber\\
&=
\begin{bmatrix}
\sst 2 &\sst  0\\
\sst 0 &\sst  -1/2
\end{bmatrix}
\hspace{-4pt}
\begin{bmatrix}
\sst 1 &\sst  (7 - z^{-1})/16\\
\sst 0 &\sst  1
\end{bmatrix}
\hspace{-4pt}
\begin{bmatrix}
\sst 1 &\sst  0\\
\sst -2 &\sst  1
\end{bmatrix}
\hspace{-4pt}
\begin{bmatrix}
\sst 1 &\sst  0\\
\sst 0 &\sst  z^{-1}
\end{bmatrix}
\hspace{-4pt}
\begin{bmatrix}
\sst 1 &\sst  -1/2\\
\sst 0 &\sst  1
\end{bmatrix}
\hspace{-4pt}
\begin{bmatrix}
\sst 0 &\sst  1\\
\sst 1 &\sst  0
\end{bmatrix}. 
\end{align*}

\subsubsection{Extracting Diagonal Delay Matrices}\label{sec:Study:Other:Delay}
CCA schemes that mix row and column updates  are also possible, but a more significant ability  will be a ``generalized polynomial division''  technique, which we now demonstrate, for  factoring off diagonal delay matrices at \emph{arbitrary} points in the  process, rather than  waiting for  $R_0$ and $R_1$ to have a nontrivial gcd.  
Divide in column~0 of LGT(5,3) to get a factorization of the form~\eqref{LGT_left_update}, but this time  generate a remainder divisible by $z^{-1}$ when dividing $F_0$ into $E_0$.  Killing the highest-order term in $E_0$ with $S(z) \leftarrow z^{-1}/4$ leaves
\begin{align*}
R_0(z) &= E_0(z) - F_0(z)\,(z^{-1}/4) 
= (-1+7z^{-1})/8 . 
\end{align*}
Instead of  subtracting 7/4 from $S(z)$ to get $R_0=-1$ as  in~\eqref{LGT_col0_step0_div},   add 1/4  to \emph{kill} the constant term, $S(z) \leftarrow (z^{-1} + 1)/4$.  This leaves
$R_0(z) = z^{-1}$, which satisfies $\deg(R_0)  < \deg(F_0) + 1$.
The motivation, generating a remainder  divisible by $z^{-1}$,  is new but the arithmetic is comparable to how Daubechies and Sweldens~\cite[\S 7.8]{DaubSwel98} generated a linear phase  lifting factorization for a {unimodular} linear phase  filter bank using the   {Laurent}  polynomial EEA.

Now set $R_1 \leftarrow E_1 - F_1S = 0$.  Interestingly, $R_1$ also happens to be divisible by $z^{-1}$; this allows the resulting  factorization to be written
\begin{align}
\mathbf{H}(z) 
&=
\begin{bmatrix}
1 & S\\
0 & 1\vspace{-1pt}
\end{bmatrix}
\hspace{-2pt}
\begin{bmatrix}
R_0 & R_1\\
F_0 & F_1\vspace{-1pt}
\end{bmatrix}
=
\begin{bmatrix}
1 & (1+ z^{-1})/4\\
0 & 1\vspace{-1pt}
\end{bmatrix}
\hspace{-3pt}
\begin{bmatrix}
z^{-1} & 0\\
0 & 1\vspace{-1pt}
\end{bmatrix}
\hspace{-3pt}
\begin{bmatrix}
1 & 0\\
 -(1 + z^{-1})/2	&  1 \vspace{-1pt}
\end{bmatrix}. \label{LGT_0:0,1}
\end{align}
This is the causal linear phase lifting factorization~\eqref{LGT_EEA_col1} so we have obtained a different factorization than~\eqref{LGT_col0_factorization}, though it is not a \emph{new} factorization.

%


\section{Linear Diophantine Equations}\label{sec:LDE}
We now focus on the theory behind the CCA  to understand the differences between unimodular and causal lifting.  
Given $a,b,$ and $c$, a \emph{linear Diophantine equation}  (LDE) in  unknowns  $x$ and $y$ is an  equation of the form
\begin{equation}\label{LDE}
ax + by = c.
\end{equation}
Indeterminate equations of this type have been  studied over the integers (albeit not using modern notation) at least as far back as Diophantus of Alexandria.  The Indian astronomer Aryabhata (fifth--sixth centuries C.E.) and his colleagues and successors found the general solution to~(\ref{LDE}) over the integers using the Euclidean Algorithm. Indeed, judging from  van der Waerden~\cite[Chapter~5]{Waerden:83:Geometry-algebra-ancient}, it appears that the general solution found by  these ancient Indian scholars consisted of the integer version of the  Lifting Theorem (Corollary~\ref{cor:AbsLiftingTheorem}).

\subsection{Factorization in Commutative Rings}\label{sec:LDE:Factorization}
While we are mainly interested in LDEs over  rings of (causal) polynomials and (noncausal) Laurent polynomials, we first determine which aspects of lifting   follow from  general factorization theory and thus do not distinguish between the causal and noncausal cases.
We follow standard terminology for commutative rings $\mathcal{R}$  \cite{MacLaneBirkhoff67,Jacobson74,Hungerford74,Herstein75}. Divisibility of $b$ by $a\neq 0$ is denoted $a\divides b$; $a$ and $b$ are  \emph{associates} if $a\divides b$ and $b\divides a$.    
A multiplicative identity for   $\mathcal{R}$ is denoted 1. 
A \emph{unit}  is any element $u$ with a multiplicative inverse,  $u^{-1}$, such that $uu^{-1}=1$.
A nonzero nonunit, $c$, is \emph{irreducible} if its only divisors are units and associates.
A subset $A\subset\mathcal{R}$ is \emph{coprime} if the only common divisors of all $a_i\in A$ are  units.  
An \emph{integral domain} is a commutative ring with  identity that contains no \emph{zero divisors:} nonzero elements $a,b$ for which $ab=0$.
%
%

\subsubsection{Unique Factorization Domains}\label{sec:LDE:Factorization:UFD}
This is the least restrictive relevant class of  commutative rings.
A \emph{unique factorization domain} is any integral domain in which each nonzero nonunit can be factored into irreducibles that are ``unique modulo associates.''   This means that if $a = \Pi a_i = \Pi b_j$ are two (finite) factorizations of $a$  into irreducibles then each $a_i$ is an associate of a distinct $b_{j(i)}$, a requirement satisfied by polynomial rings.  A common divisor $h$ for a subset $A\subset\mathcal{R}$  is a \emph{greatest common divisor} (gcd) of $A$ if all common divisors of $A$ necessarily divide $h$.  
In  unique factorization domains every finite subset with a nonzero element  has a gcd~\cite[Theorem~III.3.11(iii)]{Hungerford74}.   Gcds are not unique, and we  write $h=\gcd(A)$ if $h$ is any gcd of $A$.  
The next two  lemmas are straightforward exercises.  
%
\begin{lemma}\label{lem:GCDs}
Let $\mathcal{R}$ be  a unique factorization domain and let $A$ be a finite subset with common divisior $d$:
$a_i = \tilde{a}_i d$ for all $a_i\in A$. 
Then $d=\gcd(A)$ if and only if the set of all  $\tilde{a}_i$ is coprime.
\end{lemma}
%

%

%
\begin{lemma}\label{lem:Coprimeness}
Let $\mathcal{R}$ be a unique factorization domain and let $a,\,b\in\mathcal{R}$, $a\neq 0$. Then $a$ and $b$ are coprime if and only if, for all $c\in\mathcal{R}$, $a\divides bc$ implies $a\divides c$.
\end{lemma}
\begin{definition}\label{defn:Homog_LDE}
Let $\mathcal{R}$ be a unique factorization domain with $a,b\in\mathcal{R}$ not both zero and consider an LDE~(\ref{LDE}) involving $a,b$.  
If  $c$ in~(\ref{LDE}) is nonzero then we call~(\ref{LDE})  an \emph{inhomogeneous LDE for $a,b$ with inhomogeneity $c$}.  
If  $c$ in~(\ref{LDE}) is zero,
\begin{equation}\label{Homog_LDE}
ax + by = 0,
\end{equation}
then~\eqref{Homog_LDE} is called the \emph{homogeneous LDE for $a,b$}.

Let $h\eqdef\gcd(a,b)$ so that, by Lemma~\ref{lem:GCDs}, $a=\tilde{a}h$ and $b=\tilde{b}h$ with $\tilde{a}$ and $\tilde{b}$ coprime. The \emph{reduced homogeneous LDE for $a,b$} is
\begin{equation}\label{reduced_homog_LDE}
\tilde{a}x + \tilde{b}y = 0.
\end{equation}
\end{definition}
%

%
\begin{lemma}\label{lem:HLDE_soln_sets}
If $\mathcal{R}$ is a unique factorization domain with  $a,b$ not both zero
then~\eqref{Homog_LDE} and \eqref{reduced_homog_LDE} have identical solutions.
\end{lemma}

We can now characterize the solution sets to homogeneous and inhomogeneous LDEs.
\begin{theorem}[Homogeneous LDE Theorem]\label{thm:Homog_LDE}
Let $\mathcal{R}$ be a unique factorization domain with   $a,b$ not both zero, $h\eqdef\gcd(a,b)$,  $a=\tilde{a}h$ and $b=\tilde{b}h$.  Conditions (\ref{HC_1}) and (\ref{HC_2}) below are equivalent.
\begin{enumerate}[label={\roman*)},ref={\roman*}]		
\item\label{HC_1}  The pair $(x,y)$ satisfies~(\ref{Homog_LDE}).

\item\label{HC_2}  There exists a unique $s\in\mathcal{R}$ with $x=s\tilde{b}$ and $y=-s\tilde{a}$.
\end{enumerate}
If  (\ref{HC_1}) and (\ref{HC_2}) hold then the following are equivalent.
\begin{enumerate}[label={\roman*)},ref={\roman*}]	
\addtocounter{enumi}{2}						
\item\label{HC_3}  $x$ and $y$ are coprime.
\item\label{HC_4}  $s$ is a unit.
\end{enumerate}
\end{theorem}
\begin{proof}
(\ref{HC_1}$\Rightarrow$\ref{HC_2})  Let $(x,y)$ satisfy~(\ref{Homog_LDE}); by Lemma~\ref{lem:HLDE_soln_sets} $(x,y)$ also satisfies~\eqref{reduced_homog_LDE}.
Suppose $\tilde{a}=0$; $\tilde{a}$ and $\tilde{b}$ are coprime by Lemma~\ref{lem:GCDs} so coprimality implies that $\tilde{b}$ is a unit.  By~\eqref{reduced_homog_LDE},  $y=0$ and the unique solution to~(\ref{HC_2}) is  $s=xb^{-1}$.  Clause~(\ref{HC_2}) is similarly satisfied if $\tilde{b}=0$ so  assume   $\tilde{a},\tilde{b}\neq 0$.
We have $\tilde{a}\divides \tilde{b}y$ by~(\ref{reduced_homog_LDE}) so $\tilde{a}\divides y$ by Lemma~\ref{lem:Coprimeness}.  Thus, $y=r\tilde{a}$ for some uniquely determined $r\in\mathcal{R}$.  Similarly, $x=s\tilde{b}$ for a unique $s\in\mathcal{R}$.  One can therefore write~(\ref{reduced_homog_LDE}) as
$0 = \tilde{a}s\tilde{b} + \tilde{b}r\tilde{a} = \tilde{a}\tilde{b}(s+r).$
This implies  $r=-s$ 
since $\tilde{a}\tilde{b}\neq 0$,  proving~(\ref{HC_1}$\Rightarrow$\ref{HC_2}). 

(\ref{HC_2}$\Rightarrow$\ref{HC_1})  $(s\tilde{b},-s\tilde{a})$  satisfies~\eqref{reduced_homog_LDE} and therefore satisfies~(\ref{Homog_LDE}) by Lemma~\ref{lem:HLDE_soln_sets}.

Next, assume  that (\ref{HC_1}) and (\ref{HC_2}) hold.

(\ref{HC_3}$\Rightarrow$\ref{HC_4})  
$s$ is a common divisor of $x$ and $y$ by~(\ref{HC_2}) so if $x$ and $y$ are coprime then $s$ must be a unit, proving (\ref{HC_3}$\Rightarrow$\ref{HC_4}).  

(\ref{HC_4}$\Rightarrow$\ref{HC_3})  Let $c$ be a common divisor of $x=s\tilde{b}$ and $y=-s\tilde{a}$; if $s$ is a unit then   $c$ is also a common divisor of $\tilde{a}$ and $\tilde{b}$.  By Lemma~\ref{lem:GCDs} $\tilde{a}$ and $\tilde{b}$ are coprime  so $c$ must be a unit, implying  $x$ and $y$ are coprime.  
\qed
\end{proof}
%

%
\begin{corollary}[Abstract Lifting Theorem; cf.~\cite{Herley93,VetterliHerley:92:Wavelets-filter-banks,Sweldens96}]\label{cor:AbsLiftingTheorem}
Let $\mathcal{R}$ be a unique factorization domain with $a,b,c\in\mathcal{R}$, $a$ and $b$ not both zero.  Let $h\eqdef\gcd(a,b)$, $a=\tilde{a}h$, $b=\tilde{b}h$, and let $(x,y)$ satisfy~(\ref{LDE}).  Then $(x',y')$ is another solution to~(\ref{LDE})  if and only if  $x' = x+s\tilde{b}$ and $y' = y-s\tilde{a}$ for some  $s\in\mathcal{R}$.  
\end{corollary}
\begin{proof}
(Sufficiency)  If both $(x',y')$ and $(x,y)$ satisfy~(\ref{LDE})  then their difference,  $(x'-x,\,y'-y)$, satisfies the \emph{homogeneous} LDE~(\ref{Homog_LDE}) 
%
%
so Theorem~\ref{thm:Homog_LDE}(\ref{HC_1}$\Rightarrow$\ref{HC_2}) provides $s\in\mathcal{R}$ such that  $x' - x = s\tilde{b}$ and $y' - y = -s\tilde{a}$.  

(Necessity)  Let $(x,y)$ satisfy~(\ref{LDE}) and suppose that $x' = x+s\tilde{b}$, $y' = y-s\tilde{a}$ for some $s\in\mathcal{R}$. Theorem~\ref{thm:Homog_LDE}(\ref{HC_2}$\Rightarrow$\ref{HC_1}) implies that the pair $(x'-x,\,y'-y)$ satisfies~(\ref{Homog_LDE}) so, by bilinearity, the pair $(x',y')$ also satisfies~(\ref{LDE}). 
\qed\end{proof}
%

\rem  The Abstract Lifting Theorem (Corollary~\ref{cor:AbsLiftingTheorem}) provides all solutions to~(\ref{LDE})  in terms of the homogeneous solution set and any one   particular  solution.   
The homogeneous solution set $\{(s\tilde{b},\,-s\tilde{a}) \colon s\in\mathcal{R}\}$  is {never} just $\{(0,0)\}$, so LDEs over unique factorization domains  \emph{never} have  unique solutions.  
{Finding} a particular solution requires more   work, but unlike~\cite[Fact~4.1]{VetterliHerley:92:Wavelets-filter-banks}, which appeals to  the polynomial Euclidean Algorithm~\cite{Blahut:87:Fast-Algorithms}, we now  prove existence of inhomogeneous solutions  via more abstract methods by specializing the rings $R$ under consideration a bit more.

\subsubsection{Principal Ideal Domains}\label{sec:LDE:Factorization:PID}
The subset $\mathcal{I\subset R}$ is a \emph{principal ideal of $R$} if it is  generated by a single element, 
\begin{equation}\label{principal_ideal}
\mathcal{I} = (a) = a\mathcal{R} \eqdef \{ar\colon r\in\mathcal{R}\}.
\end{equation}
The  ideal generated by a  finite set of elements, $A\eqdef\{a_i\}_{i=0}^n$, is the subset
\begin{align}
\mathcal{I} = (A) & =  (a_0)+\cdots+(a_n) 
	 =  \{a_0 r_0+\cdots+a_n r_n\colon r_i\in\mathcal{R}\}.    \label{finitely_generated_ideal}
\end{align}
\emph{Principal ideal domains} are integral domains in which every ideal  is principal.  A classical (but nontrivial) fact  is that  principal ideal domains are unique factorization domains \cite[Theorem~III.3.7]{Hungerford74}. We need the following  basic result.

\begin{lemma}\cite[Theorem~III.3.11]{Hungerford74}\label{lem:GCDs_in_PIDs}
In a principal ideal domain  any  finite subset $A$ with a nonzero element has a gcd, and $h$ is a gcd for $A$ if and only if $(A)=(h)$.
\end{lemma}
%
%
%
%
%

%
\begin{theorem}[Abstract Bezout Theorem]\label{thm:Inhomog_LDE}
Let $\mathcal{R}$ be a principal ideal domain with $a,b,c\in\mathcal{R}$, $a$ and $b$ not both zero, $h\eqdef\gcd(a,b)$.  There exists a solution $(x,y)$ to the  LDE~(\ref{LDE})  for inhomogeneity $c$ if and only if  $h\divides c$.
\end{theorem}

\begin{proof}
If~(\ref{LDE}) has a solution then \emph{any} common divisor of  $a$ and $b$ divides $c$.
Conversely,  if $h\divides c$ then, by~\eqref{principal_ideal} and Lemma~\ref{lem:GCDs_in_PIDs},
$c\in(h)=(\{a,b\}).$  By~(\ref{finitely_generated_ideal}) this means $c\in (a)+(b)$ so there exist $x,y\in\mathcal{R}$ such that $c=ax+by$.  
\qed\end{proof}
%


\subsection{Factorization in Euclidean Domains}\label{sec:LDE:Euclidean}
Per~\cite{Bris:10:GLS-I}, the goal  of lifting  is usually to find  factorizations  that are ``size-reducing''  in some sense.  While the  matrix polynomial order was shown to be a useful measure of ``size'' in~\cite{Bris:10:GLS-I,Bris:10b:GLS-II} for linear phase liftings, that approach does not work with more general  lifting factorizations.  
Instead,  the present paper  exploits the strong factorization theory for  scalar polynomials.  Their natural ``size'' function  leads us to Euclidean domains~\cite[Definition~III.3.8]{Hungerford74}, which are automatically principal ideal domains~\cite[Theorem~III.3.9]{Hungerford74}.
\begin{definition}\label{defn:EuclideanDomain}
A \emph{Euclidean domain}, $\mathcal{R}$, is an integral domain equipped with a \emph{Euclidean size function}, 
$\sigma:\mathcal{R}\backslash\{0\}\rightarrow\mathbb{N}\eqdef\{0,1,2,\ldots\}$, 
that satisfies the following two axioms for  $a,b\in\mathcal{R}$.
\begin{enumerate}[label={\roman*)},ref={\roman*}]		
\item \label{EucDom_1}  
(Monotonicity)  If $a,b\neq 0$ then $\sigma(a)\leq\sigma(ab)$.
\item \label{EucDom_2}  
(Division algorithm)  If $b\neq 0$ then there exist $q,r\in\mathcal{R}$ such that $a=qb+r$, where either $r= 0$ or  $\sigma(r)<\sigma(b)$.
\end{enumerate}
%
\end{definition}

\subsubsection{Examples and Special Properties}\label{sec:LDE:Euclidean:Examples}
\begin{example}\label{exmp:Polynomials}
$\mathbb{F}[\zeta]$ is the ring of  polynomials in $\zeta$, $f(\zeta)=\sum_{i= 0}^{\deg(f)}f_i\zeta^i$, with coefficients $f_i\in\mathbb{F}$ for some field $\mathbb{F}$ (e.g., the field of real or complex numbers).  Let 
\begin{equation}\label{degree_size_fn}
\sigma(f)\eqdef\deg(f),\mbox{ the polynomial degree of $f$.}
\end{equation}
$\mathbb{F}[\zeta]$  with polynomial  division and the size function~(\ref{degree_size_fn}) is a Euclidean domain.  Moreover,  division of polynomials yields a \emph{unique} quotient and remainder satisfying axiom~(\ref{EucDom_2}) \cite[Theorem~IV.21]{MacLaneBirkhoff67}, \cite[Theorem~III.6.2]{Hungerford74}.
The units in $\mathbb{F}[\zeta]$ are the nonzero constant polynomials, which have size (degree) zero.  
As is commonly done in algebra,  we set  $\deg(0)\eqdef -\infty$.
With this convention, the  degree function enjoys two additional important properties.
\begin{align}
&\mbox{Homomorphism property of $\deg(f)$:\ \ }	\deg(fg) = \deg(f) + \deg(g). \label{deg_homomorphism}\\
&\mbox{Max-additive bound on $\deg(f)$:\ \ } \deg(f+g) \leq \max\{\deg(f),\, \deg(g)\}. \label{deg_add_bound}
\end{align}
\end{example}

%

\begin{example}\label{exmp:LaurentPolys}
Write Laurent polynomials, $f\in\mathbb{F}[\zeta,\,\zeta^{-1}]$, over a field $\mathbb{F}$ as 
\begin{align}
f(\zeta) = \sum_{i=m}^n f_i \zeta^i, &\mbox{\ where $f_m, f_n\neq 0;\;-\infty<m\leq n<\infty$, and}\nonumber  \\
\Lord(f) &\eqdef n-m \geq 0, \text{\ the \emph{Laurent order} (or ``length'') of  $f$.}\label{L_order}
\end{align}
$\mathbb{F}[\zeta,\,\zeta^{-1}]$ equipped with  $\sigma(f)\eqdef\Lord(f)$ and   Laurent polynomial division satisfies the axioms of a Euclidean domain, but the quotient and remainder satisfying axiom~(\ref{EucDom_2}) are not  \emph{unique} in the Laurent case.  For instance, if 
$a(\zeta) \eqdef 1 + \zeta + \zeta^2$ and $b(\zeta) \eqdef 1 + \zeta$
then the unique solution to~\eqref{EucDom_2} over the polynomials $\mathbb{F}[\zeta]$ is
\[  q(\zeta) = \zeta,\;\; r(\zeta)= 1;\;\;\deg(r)<\deg(b).  \]
This is also a solution over the Laurent polynomials, but so is
\[  q(\zeta) = 1,\;\; r(\zeta)= \zeta^2;\;\;\Lord(r)<\Lord(b).  \]

The units in $\mathbb{F}[\zeta,\,\zeta^{-1}]$ are the nonzero Laurent monomials, which have  Laurent order zero.  
As with the polynomial degree, we define $\Lord(0)\eqdef -\infty$.
The Laurent order satisfies  the homomorphism property~(\ref{deg_homomorphism}),
\begin{equation}\label{L_ord_homomorphism}
\Lord(fg)  =  \Lord(f) + \Lord(g).
\end{equation}
It does \emph{not} satisfy~(\ref{deg_add_bound}), though: 
$\Lord(\zeta + \zeta^2) = 1$ but $\Lord(\zeta^2) = \Lord(\zeta) = 0.$
This  shows that  the Laurent polynomials have ``too many units'' since each term in a Laurent polynomial is a Laurent unit, of size 0.  
\end{example}
%


\section{Factorization in Polynomial Rings}\label{sec:Polynomial}
The max-additive bound~(\ref{deg_add_bound}) implies uniqueness of quotients and remainders in polynomial division~\cite{MacLaneBirkhoff67,Hungerford74}.  In fact,  division  in  a Euclidean domain produces unique quotients and remainders if and \emph{only if} the size function satisfies an inequality like~(\ref{deg_add_bound})~\cite[Proposition~II.21]{Coppel:06:Number-Theory}.  
We will see that~(\ref{deg_add_bound})  also yields unique ``degree-reducing'' solutions to \emph{polynomial} LDEs, implying unique  causal lifting  results  that do \emph{not} hold over the Laurent (noncausal) polynomials.

\subsection{Size-Reducing Solutions  to LDEs}\label{sec:Polynomial:SizeReducing}
We  introduce  terminology for ``size-reducing'' solutions to LDEs in  Euclidean domains, 
generalizing Definition~\ref{defn:CausalComplement}.  
The size-reducing concept is important enough in the context of lifting factorization to warrant precise definitions, which seem to be lacking in the algebra literature.

\begin{definition}[Size-Reducing Solutions]\label{defn:SizeReducing}
Let $\mathcal{R}$ be  a Euclidean domain whose size function satisfies $\sigma(0)=-\infty$ and the homomorphism property,
\begin{equation}\label{size_homo}
\sigma(ab) = \sigma(a) + \sigma(b).
\end{equation}
Let $a,b,c\in\mathcal{R}$ with $a$ and $b$ not both zero. Let $h\eqdef\gcd(a,b)$ and assume $h\divides c$.   
A solution $(x,y)$  to
\begin{equation}\label{Inhomog_LDE}
ax + by = c 
\end{equation}
will be called \emph{size-reducing in $a\neq 0$} (resp., \emph{size-reducing in $b\neq 0$}) if 
\begin{align}
\sigma(y) < \sigma(a) - \sigma(h)\quad\text{(resp., $\sigma(x) < \sigma(b) - \sigma(h)$).}
\end{align}
%
%
\end{definition}

\rem By~\eqref{size_homo}, the ``correction'' terms  $\sigma(h)$  ensure that  $\sigma(a) - \sigma(h)$ and $\sigma(b) - \sigma(h)$ are invariant under cancellation of common divisors for $a$, $b$, and $c$ in~\eqref{Inhomog_LDE}, a manipulation that leaves the solution set unchanged.
We now show that   size-reducing solutions in $a$ and  $b$ may or may not agree.

\begin{example}\label{exmp:two_unique_solns}
Let  $\mathcal{R}=\mathbb{F}[\zeta]$ for some field $\mathbb{F}$, and define 
\[  a(\zeta) \eqdef 1 + \zeta,\;\; b(\zeta) \eqdef 1,\;\; c(\zeta) \eqdef \zeta,\;\; h\eqdef\gcd(a,b)=1, \mbox{ with } \deg(h)=0. \]
The pair $(x,y)=(1,-1)$ is the unique solution  to~(\ref{Inhomog_LDE}) that is degree-reducing in $a$,  $\deg(y)<\deg(a)$, while   $(0,\zeta)$ is the unique degree-reducing solution in $b$.
\end{example}
\begin{example}\label{exmp:one_common_soln}
Now redefine $a(\zeta)$ slightly: 
\[  a(\zeta) \eqdef 1 + \zeta^2,\;\; b(\zeta) \eqdef 1,\;\; c(\zeta) \eqdef \zeta,\;\; h\eqdef\gcd(a,b)=1.  \]
The pair $(x,y)=(0,\zeta)$ from Example~\ref{exmp:two_unique_solns} is still the unique solution to~\eqref{Inhomog_LDE} that is degree-reducing  in $b$, but  it is now the unique solution that is degree-reducing in $a$, too. 
If we reinterpret~\eqref{Inhomog_LDE} as a problem over the \emph{Laurent} polynomials, however, then  a second solution that is  size-reducing in $a$ is
\[ (x,y)=(\zeta^{-1},\, -\zeta^{-1}),\;\;  \Lord(y) = 0 < \Lord(a) = 2. \]
\end{example}

\subsection{LDEs in Polynomial Domains}\label{sec:Polynomial:LDEs}
We now show, using  inequality~(\ref{deg_add_bound}),  that an LDE over a \emph{polynomial} domain has exactly one solution that is degree-reducing in $a\neq 0$ or $b\neq 0$.    It is possible, as in Example~\ref{exmp:two_unique_solns},  to have  different  degree-reducing solutions in $a$ and in $b$ or, as in Example~\ref{exmp:one_common_soln}, to have one solution that is degree-reducing in both $a$ and $b$. The question of  which alternative holds is answered  by Theorem~\ref{thm:LDDRT}.
\begin{lemma}\label{lem:LDE_uniqueness}
Let $a,b,c\in\mathbb{F}[\zeta]$ with $a$ and $b$ not both zero.  Let $h\eqdef\gcd(a,b)$ and assume  $h\divides c$.  
Let $(x,y)$ be a solution to~\eqref{Inhomog_LDE}, and suppose $(x,y)$ is degree-reducing in $a$ (resp., $b$).
If $(x',y')$ is another solution to~(\ref{Inhomog_LDE}) that is degree-reducing in $a$ (resp., $b$) then $x'=x$ and $y'=y$.
\end{lemma}

\begin{proof}   Assume $(x,\,y)$ and  $(x',\,y')$ are both degree-reducing in  $a\neq 0$ (the proof is similar if both  are degree-reducing in  $b$),
\begin{equation}\label{deg_reducing_in_b}
\deg(y),\;\deg(y')<\deg(a) - \deg(h). 
\end{equation}
Factor $h$ out of $a,b,$ and $c\,$:\,
$a=h\tilde{a}$,  $b=h\tilde{b}$,  and $c=h\tilde{c}$,
where~(\ref{deg_homomorphism}) implies
\begin{equation}\label{deg_of_h_times_a_tilde}
\deg(a) = \deg(h)+\deg(\tilde{a}),\text{ etc.}
\end{equation}
Let $x''\eqdef x-x'$ and $y''\eqdef y-y'$; then $(x'',\,y'')$ satisfies the corresponding reduced homogeneous LDE~(\ref{reduced_homog_LDE}), $\tilde{a}x''+\tilde{b}y''=0$.  Theorem~\ref{thm:Homog_LDE} supplies  $s$ such that $x''=s\tilde{b}$ and $y''=-s\tilde{a}$.  
By inequality~(\ref{deg_add_bound}), assumption~(\ref{deg_reducing_in_b}), and~(\ref{deg_of_h_times_a_tilde})  
%
\begin{align*}
\deg(s\tilde{a}) \,=\, \deg(y'') \,
&\leq \,\max\{\deg(y),\deg(y')\} \\
&< \,\deg(a) - \deg(h) \,=\, \deg(\tilde{a}) \,\leq\, \deg(s\tilde{a})\text{\ \  if $s\neq 0$,}
\end{align*}
which is a contradiction unless  $s=0$. Thus,
$x-x'   =  x''  =  s\tilde{b}   =  0$ and $y-y'   =  y''  =  -s\tilde{a}  =  0$,
proving uniqueness of degree-reducing solutions. 
\qed
\end{proof}
%

The main result in this section, the Linear Diophantine Degree-Reduction Theorem (or LDDRT, Theorem~\ref{thm:LDDRT}), is a strengthened version of \cite[Theorem~4.10(iii)]{GathenGerhard:13:Modern-Computer-Algebra} and Bezout's Theorem for polynomials, 
e.g.~\cite[Theorem~6.1.1]{Daub92}.
It uses the Abstract Lifting Theorem (Corollary~\ref{cor:AbsLiftingTheorem}), the Abstract Bezout Theorem (Theorem~\ref{thm:Inhomog_LDE}),  Lemma~\ref{lem:LDE_uniqueness}, and the  polynomial division algorithm 
to characterize existence and  uniqueness of \emph{degree-reducing} solutions to inhomogeneous  LDEs over \emph{polynomial} domains.  Theorem~\ref{thm:LDDRT}(\ref{DegRed_1})  is also used to prove  Corollary~\ref{cor:GDT}, which the author has been unable to find in the literature.
\begin{theorem}[Linear Diophantine Degree-Reduction Theorem]\label{thm:LDDRT}
Let  $a,b,c\in\mathbb{F}[\zeta]$, $a$ and $b$ not both zero.  Let $h\eqdef\gcd(a,b)$; assume  $h\divides c$.  
\begin{enumerate}[label={\roman*)},ref={\roman*}]		
\item\label{DegRed_1}
If $a\neq 0$ then there exists a unique solution, $(x,y)$, to the LDE
\begin{equation}\label{LDE_again_again}
ax +by = c  
\end{equation}
that is degree-reducing in $a$,
\begin{equation}\label{deg_reducing_a}
\deg(y) < \deg(a) - \deg(h).
\end{equation}
\item\label{DegRed_2}
If $b\neq 0$ then there exists a unique solution, $(x',y')$, to~(\ref{LDE_again_again}) that is degree-reducing in $b$,
\begin{equation}\label{deg_reducing_b}
\deg(x') < \deg(b) - \deg(h).
\end{equation}
\item\label{DegRed_3}
Let $a,b\neq 0$ and let $(x,y)$ and $(x',y')$ be the  solutions  in clauses~(\ref{DegRed_1}) and~(\ref{DegRed_2}).  These two solutions are the same,  $x= x'$ and $y= y'$, if and only if
\begin{equation}\label{c_small}
\deg(c) < \deg(a)+\deg(b)-\deg(h)\,.
\end{equation}
%
\end{enumerate}
\end{theorem}
\begin{proof}    
As in the proof of Lemma~\ref{lem:LDE_uniqueness}, if $a$ and $b$ are not coprime then we  factor $h$ out of $a$, $b$, and $c$ in~(\ref{LDE_again_again}),  leaving an equivalent LDE with $\tilde{a}$ and $\tilde{b}$ coprime,
\begin{equation}\label{coprime_LDE}
\tilde{a}x+\tilde{b}y=\tilde{c}.
\end{equation}
By~(\ref{deg_of_h_times_a_tilde}) the degree-reducing conditions~\eqref{deg_reducing_a} (resp., \eqref{deg_reducing_b}) are equivalent to 
\begin{align}\label{deg_reducing_a_tilde}
\deg(y)<\deg(\tilde{a})\quad (\text{resp.,\ }\deg(x')<\deg(\tilde{b})).
\end{align} 
Formula~(\ref{deg_of_h_times_a_tilde})  also implies that~(\ref{c_small}) 
is equivalent to
\begin{align}
\deg(\tilde{c})&< \deg(\tilde{a})+\deg(\tilde{b})\,.\label{c_tilde_small}
\end{align}
Since $h\divides c$, by Theorem~\ref{thm:Inhomog_LDE} there exists a solution $(x^*,y^*)$ to~(\ref{LDE_again_again}) and~(\ref{coprime_LDE}).

(\ref{DegRed_1})  Divide $\tilde{a}\neq 0$ into $y^*$ to get $q$ and $y$ satisfying
\begin{equation}\label{y_star_over_a}
y^* = q\tilde{a}+y,\quad \deg(y) < \deg(\tilde{a}).
\end{equation}
Let $x\eqdef x^* + q\tilde{b}$; then $(x,y) = (x^*,y^*) + (q\tilde{b},-q\tilde{a})$.
By Corollary~\ref{cor:AbsLiftingTheorem} $(x,y)$ is also a solution to~(\ref{LDE_again_again}).  Inequality~(\ref{y_star_over_a}) is precisely~(\ref{deg_reducing_a_tilde}) so $(x,y)$ is degree-reducing in $a$ and uniqueness  follows from Lemma~\ref{lem:LDE_uniqueness}.

(\ref{DegRed_2})  The proof of clause~(\ref{DegRed_2}) is similar to the proof of~(\ref{DegRed_1}).

(\ref{DegRed_3})  
Let $(x,y)$ and $(x',y')$ be the unique degree-reducing solutions given in clauses~(\ref{DegRed_1}) and~(\ref{DegRed_2}), respectively,
\begin{align}
\deg(y) < \deg(\tilde{a})\text{ and }\deg(x') < \deg(\tilde{b}).\label{coprime_deg_reducing_a}
\end{align}
Suppose  $(x,y)=(x',y')$; applying~(\ref{deg_add_bound}) and~(\ref{deg_homomorphism}) to~\eqref{coprime_LDE}  proves~\eqref{c_tilde_small},
\begin{align*}
\deg(\tilde{c})  \leq  \max\{\deg(\tilde{a}x),\,\deg(\tilde{b}y)\} 
&=  \max\{\deg(\tilde{a})+\deg(x),\,\deg(\tilde{b})+\deg(y)\}  \\
&<  \deg(\tilde{a})+\deg(\tilde{b})\text{ by~(\ref{coprime_deg_reducing_a}) since } x'=x.
\end{align*}  

Conversely, assume~\eqref{c_tilde_small}.  By~(\ref{deg_homomorphism}) and~(\ref{coprime_deg_reducing_a}) we  have
\begin{equation}\label{deg_by_bound}
\deg(\tilde{b}y) = \deg(\tilde{b})+\deg(y) < \deg(\tilde{b})+\deg(\tilde{a}).
\end{equation}
Therefore,  by~\eqref{deg_homomorphism}, \eqref{coprime_LDE}, and~\eqref{deg_add_bound},
\begin{align*}
\deg(\tilde{a})+\deg(x)	
&=	 \deg(\tilde{a}x)
=	 \deg(\tilde{c}-\tilde{b}y) \\
&\leq  \max\{\deg(\tilde{c}),\,\deg(\tilde{b}y)\} 
<	\deg(\tilde{a}) + \deg(\tilde{b})\text{ by~(\ref{c_tilde_small}), (\ref{deg_by_bound}).}
\end{align*}
This implies 
$\deg(x) < \deg(\tilde{b}) =  \deg(b) - \deg(h)$,  which says that  $(x,y)$ is also degree-reducing in $b$.
But $(x',y')$ is the unique solution to~(\ref{LDE_again_again}) that is  degree-reducing in $b$, so  $(x,y)=(x',y')$.  
\qed
\end{proof}
%
	
\section{Generalized Polynomial  Division}\label{sec:LDE:Division}
We now generalize the  polynomial  division algorithm using the LDDRT to handle divisibility constraints on  remainders.  A specialization of this  will be used in the CCA  to factor diagonal delay matrices off of causal PR filter banks.

\subsection{Ideal-Theoretic Interpretation}  
Given polynomials $e$ and $f\neq 0$ in $\mathbb{F}[\zeta]$, the  polynomial division algorithm yields a unique quotient $q$ whose remainder, $r\eqdef e-fq$, satisfies $\deg(r)<\deg(f)$.
This forms the unique solution $(q,r)$ that is degree-reducing in $f$ for the LDE
\begin{equation}\label{classical_DA_LDE}
fq + 1r = e.
\end{equation}
%
Consequently, the coset $e + (f)$ of the ideal generated by $f$  contains a unique element, $r= e-fq$,  that satisfies  $\deg(r) < \deg(f)$, making it the unique element of minimum degree in $e + (f)$.

Theorem~\ref{thm:LDDRT} implies a far-reaching generalization of  classical polynomial division.  
Let polynomials $e$ and $f,g\neq 0$ be given such that $h\eqdef\gcd(f,g)\divides  e$.  
Theorem~\ref{thm:LDDRT}\eqref{DegRed_1} yields a unique solution $(q,p)$ that is degree-reducing in $f$ to 
\begin{equation}\label{GDA_LDE}
fq + gp = e,\mbox{ with } \deg(p) < \deg(f) - \deg(h).
\end{equation}
This makes $r \eqdef e - fq =  gp$  the  unique element of minimum degree in the coset-ideal \emph{intersection} $[e + (f)]\cap (g)$ since it is the unique element  that satisfies 
\begin{align*}
\deg(r) \,=\, \deg(p) + \deg(g) \,<\, \deg(f) - \deg(h) + \deg(g)\text{\ \ by \eqref{deg_homomorphism} and \eqref{GDA_LDE}.}
\end{align*}
This proves the following result.
\begin{corollary}[Generalized Polynomial Division Theorem]\label{cor:GDT}
Let $e,f,g\in\mathbb{F}[\zeta]$, $f,g\neq 0$. 
If $h\eqdef\gcd(f,g)$ divides $e$
then there exists a unique quotient $q$ whose remainder, $r\eqdef e-fq$, is divisible by $g$ and satisfies 
\[  \deg(r)<\deg(f) - \deg(h) +\deg(g).  \]
\end{corollary}

\subsection{A Constructive Generalized  Division Algorithm}  
In the lifting context,  $e$ and $f$ in  Corollary~\ref{cor:GDT} are the  polynomials in one row or column of a  transfer matrix, $\zeta=z^{-1}$,
and $g(\zeta)=z^{-M}$.  The  proof via the (nonconstructive) LDDRT (Theorem~\ref{thm:LDDRT}) will be replaced by a constructive proof (Theorem~\ref{thm:SGDT}) and a computational algorithm.  But first, one more definition.
\begin{definition}\label{defn:deg_red_mod_M}
Let  $e,f\in\mathbb{F}[\zeta]$, $f\neq 0$, and $M\geq 0$. Given quotient $q$,  the remainder $r\eqdef e-fq$ has (a root  at $\zeta=0$ of) \emph{multiplicity~$M$} if  $\zeta^M \divides r(\zeta)$.  We say  $r$  is \emph{degree-reducing modulo $M$} if  $r$ has multiplicity~$M$  and  satisfies
\[  \deg(r) < \deg(f) - \deg\gcd(f,\zeta^M) + M.  \]
\end{definition}

\rem  
In our usage,  ``$r$ has multiplicity~$M$''  includes the possibility that $\zeta^{M+1} \divides r$.
The classical  division algorithm~\cite[Algorithm~2.5]{GathenGerhard:13:Modern-Computer-Algebra}  provides existence and uniqueness of  quotients whose remainders  are degree-reducing modulo~0.

\begin{theorem}[Slightly Generalized  Division Theorem]\label{thm:SGDT}
Let $M\geq 0$, $e,f\in\mathbb{F}[\zeta]$,  $f\neq 0$.  If  $\gcd(f, \zeta^M) \divides e$ 
then there exists a unique quotient  $q^{(M)}$ whose remainder, $r^{(M)}\eqdef e-fq^{(M)}$,  is degree-reducing modulo $M$, 
\begin{gather}
\zeta^M \divides r^{(M)} \text{\quad and }\label{multiplicity_M}\\ 
\deg(r^{(M)}) < \deg(f) - \deg\gcd(f,\zeta^M) + M. \label{SGDA_deg_red_mod_M}
\end{gather}
\end{theorem}

\begin{proof}   Induction on $M$.

\emph{Case: $M=0$.}  This is the classical  polynomial division algorithm.

\emph{Case: $M>0$.}  Assume the theorem holds whenever  $\gcd(f, \zeta^{M-1}) \divides e$.  Suppose $e,f$ are such that $\gcd(f, \zeta^M) \divides e$; then  $\gcd(f, \zeta^{M-1}) \divides e$ so by hypothesis there exists a unique $q^{(M-1)}$ whose remainder, $r^{(M-1)}\eqdef e-fq^{(M-1)}$, satisfies
\begin{gather}
\zeta^{M-1} \divides r^{(M-1)} \text{\quad and }\label{multiplicity_M-1}\\ 
\deg(r^{(M-1)}) < \deg(f) - \deg\gcd(f,\zeta^{M-1}) + M - 1. \label{SGDA_deg_red_mod_M_minus1}
\end{gather}

$\bigl($Existence of $q^{(M)}\text{ and }r^{(M)}.)$   Note that
\begin{align}\label{deg-gcd-ineq}
\deg\gcd(f,\zeta^M) \leq \deg\gcd(f,\zeta^{M-1}) + 1.
\end{align}
Apply~\eqref{deg-gcd-ineq} to~\eqref{SGDA_deg_red_mod_M_minus1} to get a different bound on $\deg(r^{(M-1)})$,
\begin{align}
\deg(r^{(M-1)}) 
&< 
\deg(f) - \bigr(\deg\gcd(f,\zeta^{M-1}) + 1\bigr) + M\nonumber\\
&\leq 
\deg(f) - \deg\gcd(f,\zeta^{M}) + M.\label{deg_rM-1_bound}
\end{align}
If $r^{(M-1)}_{M-1}=0$   set $q^{(M)}\eqdef q^{(M-1)}$ and get  $r^{(M)}=r^{(M-1)}$. $r^{(M-1)}_{M-1}=0$ implies $\zeta^M\divides r^{(M-1)}=r^{(M)}$, yielding~\eqref{multiplicity_M}.  Since $r^{(M)}=r^{(M-1)}$, \eqref{deg_rM-1_bound} implies~\eqref{SGDA_deg_red_mod_M}, proving existence of $q^{(M)}$ and $r^{(M)}$ satisfying~\eqref{multiplicity_M}--\eqref{SGDA_deg_red_mod_M} when $r^{(M-1)}_{M-1}=0$.  

Now assume  that $r^{(M-1)}_{M-1}\neq 0$.  Define $m_f \eqdef \max\{i\geq 0:\zeta^i\divides f\}$; then
\begin{align}
f(\zeta) = \sum_{i=m_f}^{m} f_i\zeta^i\text{ where } f_{m_f}\neq 0\text{ and } m = \deg(f). \label{f_series}
\end{align}
Suppose $m_f\geq M$; then $\zeta^M=\gcd(f,\zeta^M)$ divides both $f$ and $e$ (by hypothesis) so $\zeta^M\divides r^{(M-1)}$ since $r^{(M-1)}=e - fq^{(M-1)}$, contradicting the assumption that $r^{(M-1)}_{M-1}\neq 0$.  This forces $m_f<M$, implying
\begin{align}\label{mf_eq_deggcd}
m_f = \deg\gcd(f,\zeta^M).
\end{align}
It also ensures that the following definition is a polynomial,
\begin{align}\label{def_qM}
q^{(M)}(\zeta) \eqdef q^{(M-1)}(\zeta) + f_{m_f}^{-1} r^{(M-1)}_{M-1} \zeta^{M-1-m_f}.
\end{align}
The remainder, $r^{(M)}\eqdef e-fq^{(M)}$, can now be written
\begin{align}\label{rM_formula}
r^{(M)}(\zeta) = r^{(M-1)}(\zeta) - f(\zeta) f_{m_f}^{-1} r^{(M-1)}_{M-1} \zeta^{M-1-m_f}.
\end{align}
Using~\eqref{f_series},  its $(M-1)^{th}$ (i.e., lowest-order) coefficient  is
\[  r^{(M)}_{M-1} = r^{(M-1)}_{M-1} - f_{m_f} f_{m_f}^{-1} r^{(M-1)}_{M-1} = 0.  \]
This implies $\zeta^M \divides r^{(M)}$ so~\eqref{multiplicity_M} is satisfied.
To prove~\eqref{SGDA_deg_red_mod_M} apply~\eqref{deg_add_bound} to~\eqref{rM_formula},
\begin{align}\label{rM_max_add_ineq}
\deg(r^{(M)}) \leq 
\max\bigl\{\deg(r^{(M-1)}),\,\deg(f) + M - 1 - m_f\bigr\}.
\end{align}
$\deg(r^{(M-1)})$ satisfies~\eqref{deg_rM-1_bound}  and by~\eqref{mf_eq_deggcd} the second argument in~\eqref{rM_max_add_ineq} is also strictly less than $\deg(f) - \deg\gcd(f,\zeta^{M}) + M$ so~\eqref{SGDA_deg_red_mod_M} is satisfied,
\begin{align*}
\deg(r^{(M)}) < \deg(f) - \deg\gcd(f,\zeta^{M}) + M.
\end{align*}

(Uniqueness.)  Suppose $(q',r')$ is another solution that is degree-reducing modulo $M$.
Subtract $r' = e - fq'$ from $r^{(M)} = e - fq^{(M)}$,  
\begin{align}\label{fqq_eq_rr}
\bigl(q' - q^{(M)}\bigr)f &= r^{(M)} - r' .
\end{align}
Take the degree of both sides,
\begin{align*}
\deg\bigl(q' - q^{(M)}\bigr) + \deg(f)
= 
\deg\bigl(r^{(M)} - r'\bigr)
&\leq 
\max\bigl\{\deg(r^{(M)}),\,\deg(r')\bigr\} \\
&< 
\deg(f) - \deg\gcd(f,\zeta^{M}) + M,  
\end{align*}
where the last inequality is hypothesis~\eqref{SGDA_deg_red_mod_M}.  Simplify this to get
\begin{align}
\deg\bigl(q' - q^{(M)}\bigr) &< M - \deg\gcd(f,\zeta^{M}). \label{M_minus_deg_gcd}
\end{align}
Define $m_f \eqdef \max\{i\geq 0:\zeta^i\divides f\}$ and write 
$f(\zeta) = \zeta^{m_f} f'(\zeta)$, where $\zeta\dividesnot f'$.  
If $\deg\gcd(f,\zeta^{M}) = M$ then  $q' = q^{(M)}$ by~\eqref{M_minus_deg_gcd} and we're done.

Otherwise, assume  $\deg\gcd(f,\zeta^{M}) < M$, which implies
$ m_f = \deg\gcd(f,\zeta^M).$
The  hypothesis~\eqref{multiplicity_M} for $r^{(M)}$ and $r'$ means that $\zeta^{M} \divides \bigl(r^{(M)} - r'\bigr)$, so~\eqref{fqq_eq_rr} implies
$\zeta^{M} \divides f\bigl(q' - q^{(M)}\bigr).$
Factor  $\zeta^{m_f}$ out of $f(\zeta)$ and cancel to infer that
$\zeta^{M-m_f} \divides f'\bigl(q' - q^{(M)}\bigr)$.   Since $\zeta\dividesnot f'$ this implies 
\[  \zeta^{M-m_f} \divides \bigl(q' - q^{(M)}\bigr),\text{\ \ where\ \ } M - m_f = M - \deg\gcd(f,\zeta^{M}) > 0.   \]
This contradicts~\eqref{M_minus_deg_gcd} unless $q' - q^{(M)} = 0$, proving uniqueness of $q^{(M)}$. 
\qed
\end{proof}

\rem
If $r^{(M)}$ is degree-reducing modulo~$M$ but has multiplicity~$M+k$, $k > 0$, then it  follows from~\eqref{SGDA_deg_red_mod_M} that it is also degree-reducing modulo~$M+k$, implying $q^{(M+k)}=q^{(M)}$ and  $r^{(M+k)}=r^{(M)}$.

\subsection{A Computational Generalized Division Algorithm}\label{sec:LDE:Division:Computational}
Theorem~\ref{thm:SGDT} yields an algorithm that generalizes the  classical  polynomial division algorithm. 
For technical reasons, the formal statement of the CCA 
assumes that   $f$ is coprime to $\zeta^M$, i.e., that  $f_0 \neq 0$ if  $M>0$, so we  simplify our Slightly Generalized Division Algorithm (SGDA) accordingly. 
We  also reduce the complexity of the initial loop (the classical division algorithm) a bit when $M>0$. Since we are assuming that $\deg\gcd(f,\zeta^{M})=0$, \eqref{SGDA_deg_red_mod_M} simplifies to $\deg(r)< \deg(f)+M$ so the SGDA  only needs to reduce  $\deg(r)$ by this much.  

Polynomials $f(\zeta)=\sum_{i=0}^m f_i\zeta^i$, $f_m\neq 0$,   are represented   by  coefficient vectors in bold italics, $\bsy{f}=(f_0,\ldots,f_m)$, where $\deg(\bsy{f})\eqdef \deg(f)=m$ and where
different-length  vectors are implicitly extended with zeros in sums.  Multiplication of $f(\zeta)$ by $\zeta^k$  right-shifts its coefficient vector by $k$,
\begin{equation}\label{right_shift}
(\tau_k \bsy{f})_n \eqdef \left\{
\begin{array}{cl}
0, & 0\leq n<k,   \\
f_{n-k}, & n\geq k.
\end{array}
\right.
\end{equation}
Thus,  $\deg(\tau_k\bsy{f})=\deg(\bsy{f})+k$.
The output condition $r^{(M)}_k=0$ for  $0\leq k\leq M-1$ in Algorithm~\ref{alg:SGDA}  is equivalent to \mbox{$\zeta^M \divides  r^{(M)}(\zeta)$.}
Left-arrows represent assignments to  registers (e.g., \mbox{$k\leftarrow k+1$} or  $\bsy{q}\leftarrow \bsy{0}$) that may be overwritten later.

\begin{algorithm}
\caption{\sf (Slightly Generalized Division Algorithm).}\label{alg:SGDA}
\begin{flushleft}
\textbf{Input:}  Integer $M\in\mathbb{N}$. Dividend  vector $\bsy{e}$ of degree $n$.  Divisor   $\bsy{f}\neq \mathbf{0}$ of degree $m\leq n$, with $f_0 \neq 0$ whenever  $M>0$.

\textbf{Output:}  Quotient  vector $\bsy{q}^{(M)}$ of degree $\leq n-m$.  Remainder  vector $\bsy{r}^{(M)}$ with $r^{(M)}_k=0$ for  $0\leq k\leq M-1$ satisfying the bound
$\deg({\bsy{r}^{(M)}})< \deg(\bsy{f})+M$.
\end{flushleft}

\algsetup{indent=2em, linenosize=\normalsize, linenodelimiter=.}
\begin{algorithmic}[1]
\STATE \textbf{initialize} $\bsy{q}\leftarrow \bsy{0}$\hfill(vector of  $n-m+1$ zeros)
\STATE \textbf{initialize} $\bsy{r}\leftarrow\bsy{e}$
\FOR{$(k\leftarrow n-m,\;k\geq M,\;k\leftarrow k-1)$}\label{alg:SGDA:for-}
	\IF{$r_{m+k}\neq 0$}
		\STATE $q_k\leftarrow f_m^{-1} r_{m+k}$
		\STATE $\bsy{r}\leftarrow\bsy{r} - q_k(\tau_k\bsy{f})$\hfill(renders $r_{m+k}=0$)
	\ENDIF
\ENDFOR
\FOR{$(k\leftarrow 0,\;k< M,\;k\leftarrow k+1)$}\label{alg:SGDA:for+}
	\IF{$r_k\neq 0$}
		\STATE $q_k\leftarrow f_0^{-1}r_k$
		\STATE $\bsy{r}\leftarrow\bsy{r} - q_k(\tau_k\bsy{f})$\hfill(renders $r_k=0$)
	\ENDIF
\ENDFOR
\RETURN $\bsy{q}^{(M)} = \bsy{q}$, $\bsy{r}^{(M)} = \bsy{r}$
\end{algorithmic}
\end{algorithm}
%

\rem
Line~\ref{alg:SGDA:for-}: When $M=0$ this loop implements the classical  division algorithm.   When $M>0$ it only zeros out enough high-order terms to ensure   
\begin{align}\label{deg-reduction-limit}
\deg({\bsy{r}^{(M)}})< \deg(\bsy{f})+M = m+M. 
\end{align}
If $n-m<M$ then~\eqref{deg-reduction-limit}  is satisfied by the initial condition $\deg({\bsy{r}})=\deg({\bsy{e}})=n$ so this loop is not traversed.

Line~\ref{alg:SGDA:for+}: This loop is only traversed  if $M>0$, which assumes  $f_0 \neq 0$.  It ensures that $\zeta^M \divides  r^{(M)}(\zeta)$ while preserving the bound $\deg({\bsy{r}^{(M)}})< \deg(\bsy{f})+M$.

\section{Case Study: The  Cubic B-Spline Binomial Filter Bank}\label{sec:CubicBSpline} 
CDF(7,5)  is a biorthogonal  wavelet filter bank  constructed by Cohen, Daubechies, and Feauveau 
\cite[\S 6.A]{CohenDaubFeau92}, \cite[\S 8.3.4]{Daub92}; the analog synthesis scaling function and mother wavelet generate cubic B-splines.   Reverting to  Z-transform notation ($\zeta\leftarrow z^{-1}$), the noncausal synthesis filters are
\begin{align*}
P_0(z) &= (z^2 + 4z + 6 + 4z^{-1} + z^{-2})/8,\\
P_1(z) &= (-3z^2 - 12z - 5 + 40z^{-1} - 5z^{-2} - 12z^{-3} - 3z^{-4})/32.
\end{align*}
Daubechies and Sweldens~\cite[\S 7.8]{DaubSwel98}  factored the  unimodular  PWA synthesis matrix $\mathbf{P}(z)$ into  linear phase lifting steps using the Laurent polynomial EEA. The corresponding PWA analysis matrix $\mathbf{A}(z)$ and its  factorization are
%
\begin{align}
\mathbf{A}(z) 
&=
\begin{bmatrix}
\sst (-3z + 10 - 3z^{-1})/8  &\,\sst  (3z + 5 + 5z^{-1} + 3z^{-2})/32 \\
\sst -(z + 1)/2 &\,\sst (z + 6 + z^{-1})/8
\end{bmatrix},\quad
|\mathbf{A}(z)| =  1,\nonumber\\
&=
\begin{bmatrix}
\sst 2 &\sst  0\\
\sst 0 &\sst  1/2
\end{bmatrix}
\hspace{-4pt}
\begin{bmatrix}
\sst 1 &\sst  3(1 + z^{-1})/16\\
\sst 0 &\sst  1
\end{bmatrix}
\hspace{-4pt}
\begin{bmatrix}
\sst 1 &\sst  0\\
\sst -(z + 1) &\sst  1
\end{bmatrix}
\hspace{-4pt}
\begin{bmatrix}
\sst 1 &\sst  -(1 + z^{-1})/4\\
\sst 0 &\sst  1
\end{bmatrix}.
\label{CDF75_PWA_anal}
\end{align}
As with~(\ref{LGT_GLF}), this is the unique WS group lifting 
factorization of $\mathbf{A}(z)$~\cite{Bris:10:GLS-I,Bris:10b:GLS-II}; its lifting filters require just two multipliers.
We will now factor the  7-tap/5-tap \emph{causal} PWD analysis matrix, 
\begin{align}\label{CDF75}
\mathbf{H}(z)
&=
\begin{bmatrix}
\sst  (3 + 5z^{-1} + 5z^{-2} + 3z^{-3})/32 &\,\sst  (-3 + 10z^{-1} - 3z^{-2})/8\\
\sst (1 + 6z^{-1} + z^{-2})/8 &\,\sst  -(1 + z^{-1})/2 
\end{bmatrix},\quad 
|\mathbf{H}(z)| =  -z^{-2}.
\end{align}
%

\subsection{Causal EEA Factorization in Column~1}\label{sec:CubicBSpline:EEA_Col1}  
Initialize the remainders for the EEA,
$r_0 \eqdef H_{01}(z)$ 
and
$r_1 \eqdef H_{11}(z).$ 
Using the   division algorithm, $r_0 = q_0 r_1 + r_2$ where $q_0 = (-13 + 3z^{-1})/4$, $r_2 = -2$, and
\begin{align*}
\begin{pmatrix}
r_0\\
r_1
\end{pmatrix}
=
\mathbf{M}_0
\begin{pmatrix}
r_1\\
r_2
\end{pmatrix}
\text{ for\ \ }
\mathbf{M}_0 \eqdef
\begin{bmatrix}
q_0 & 1\\
1 & 0
\end{bmatrix} .
\end{align*}
Next, $r_1 = q_1 r_2 + r_3$ where $q_1 = (1 + z^{-1})/4$,  $r_3 = 0$, and
\begin{align*}
\begin{pmatrix}
r_1\\
r_2
\end{pmatrix}
=
\mathbf{M}_1
\begin{pmatrix}
r_2\\
0
\end{pmatrix}
\text{ for\ \ }
\mathbf{M}_1 \eqdef
\begin{bmatrix}
q_1 & 1\\
1 & 0\vspace{-1pt}  
\end{bmatrix} .
\end{align*}
Now define an augmentation matrix with causal filters $a_0$ and $a_1$,
\begin{align*}
\mathbf{H}'(z) 
&\eqdef
\begin{bmatrix}
a_0 & r_0\\
a_1 & r_1\vspace{1pt}  
\end{bmatrix}
=
\mathbf{M}_0\mathbf{M}_1
\begin{bmatrix}
 0 &  r_2 \\
 -|\mathbf{H}|/r_2 &  0\vspace{-1pt}  
\end{bmatrix}\\
&=
\begin{bmatrix}
  z^{-2}(13 -  3z^{-1})/8 &\,  (-3 + 10z^{-1} - 3z^{-2})/8\\
 -z^{-2}/2 &\,  -(1 + z^{-1})/2 
\end{bmatrix} .
\end{align*}
Apply the Lifting Theorem and transform to standard causal lifting form.

\begin{align}
\mathbf{H}(z) 
&= \mathbf{H}'(z)\,\mathbf{V}(z)
= \mathbf{H}'(z) 
\begin{bmatrix}
1 & 0\\
V(z)  & 1
\end{bmatrix} ,\text{ where } V(z) = -(1 + 5z^{-1})/4,\nonumber\\
&\hspace{-1em}=
(\mathbf{M}_0\,\mathbf{J})(\mathbf{J}\,\mathbf{M}_1)\,\diag(-2,-1/2)\,\diag(1,z^{-2})\,\mathbf{J}\,\mathbf{V}(z);
\text{ use~\eqref{symbolic_DIF}, \eqref{dbl_transpose},}\nonumber\\
&\hspace{-1em}=
\begin{bmatrix}
\sst -2 &\sst  0\\
\sst 0 &\sst  -1/2
\end{bmatrix}
\hspace{-4pt}
\begin{bmatrix}
\sst 1 &\sst  (-13 + 3z^{-1})/16\\
\sst 0 &\sst  1
\end{bmatrix}
\hspace{-4pt}
\begin{bmatrix}
\sst 1 &\sst  0\\
\sst 1 + z^{-1} &\sst 1
\end{bmatrix}
\hspace{-4pt}
\begin{bmatrix}
\sst 1 & \sst 0\\
\sst 0 &\sst z^{-2}
\end{bmatrix}
\hspace{-4pt}
\begin{bmatrix}
\sst 1 &\sst  -(1 + 5z^{-1})/4\\
\sst 0 &\sst  1
\end{bmatrix}
\hspace{-4pt}
\begin{bmatrix}
\sst 0 &\sst  1\\
\sst 1 &\sst  0
\end{bmatrix} .\label{CDF75_EEA_col1}
\end{align}

Unlike what happened when factoring LGT(5,3) in Section~\ref{sec:Study:EEA:Col1}, factoring CDF(7,5) using the causal EEA in column~1 does \emph{not} produce a causal analogue of the unimodular linear phase  lifting factorization~(\ref{CDF75_PWA_anal}). Note how the entire determinantal delay of $\mathbf{H}(z)$, which was introduced via  the augmentation matrix $\mathbf{H}'(z)$, winds up in a single diagonal delay matrix, $\diag(1,z^{-2})$.

\subsection{CCA Factorization in Column~1}\label{sec:CubicBSpline:CCA_Col1}  
Initialize  $\mathbf{Q}_0(z)\eqdef \mathbf{H}(z)$; set  $E_1\leftarrow H_{01}$  and $F_1\leftarrow H_{11}$, $\deg(F_1)<\deg(E_1)$, and seek a lifting factorization, 
\begin{align}
\mathbf{Q}_0(z)
&=
\begin{bmatrix}
  (3 + 5z^{-1} + 5z^{-2} + 3z^{-3})/32 &\;  (-3 + 10z^{-1} - 3z^{-2})/8\\
 (1 + 6z^{-1} + z^{-2})/8 & -(1 + z^{-1})/2 
\end{bmatrix} \nonumber\\
&=
\begin{bmatrix}
E_0 & E_1\\
F_0 & F_1\vspace{-1pt}
\end{bmatrix}
=
\begin{bmatrix}
1 & S\\
0 & 1\vspace{-1pt}%
\end{bmatrix}
\hspace{-3pt}
\begin{bmatrix}
R_0 & R_1\\
F_0 & F_1\vspace{-1pt}%
\end{bmatrix}.\label{CDF75_col1_CCA_step0_form}
\end{align}
%
Divide $F_1$ into $E_1$ using   classical polynomial  division, $E_1 = F_1S + R_1$,  where 
\begin{align*}
S(z) = (-13 + 3z^{-1})/4\text{ and }R_1(z) = -2,\; \deg(R_1)  < \deg(F_1).
\end{align*}
Define $R_0 \eqdef E_0 - F_0S = (1 + 5z^{-1})/2$; the first factorization step is
\begin{align}
\mathbf{Q}_0(z) 
&=
\begin{bmatrix}
\sst  1 &\sst   (-13 + 3z^{-1})/4\\
\sst  0 &\sst   1%
\end{bmatrix}
\hspace{-4pt}
\begin{bmatrix}
\sst  (1 + 5z^{-1})/2 &\,\sst  -2\\
\sst (1 + 6z^{-1} + z^{-2})/8 &\,\sst  -(1 + z^{-1})/2 
\end{bmatrix}
=
\mathbf{V}_0(z)\mathbf{Q}_1(z). \label{CDF75_col1_CCA_step0_result}
\end{align}

Reset the labels $E_j\leftarrow F_j$ and $F_j\leftarrow R_j$ in $\mathbf{Q}_1(z)$,
\begin{align}\label{CDF75_col1_CCA_step1_form}
\mathbf{Q}_1(z) 
&=
\begin{bmatrix}
\sst  (1 + 5z^{-1})/2 &\,\sst  -2\\
\sst (1 + 6z^{-1} + z^{-2})/8 &\,\sst  -(1 + z^{-1})/2 
\end{bmatrix}
=
\begin{bmatrix}
F_0 & F_1\\
E_0 & E_1\vspace{-1pt}
\end{bmatrix}
=
\begin{bmatrix}
1 & 0\\
S & 1\vspace{-1pt}%
\end{bmatrix}
\hspace{-3pt}
\begin{bmatrix}
F_0 & F_1\\
R_0 & R_1\vspace{-1pt}%
\end{bmatrix}.
\end{align}
Dividing in column~1,
%
%
$E_1 = F_1S + R_1$ for $S(z) = (1 + z^{-1})/4$ with $R_1 = 0$.
Set $R_0\leftarrow E_0 - F_0S = -z^{-2}/2$; the second step is
\begin{align}\label{CDF75_col1_CCA_step1_result}
\mathbf{Q}_1(z) 
&=
\begin{bmatrix}
\sst  1 &\sst   0\\
\sst  (1 + z^{-1})/4 &\sst   1%
\end{bmatrix}
\hspace{-4pt}
\begin{bmatrix}
\sst  1 &\sst   0\\
\sst  0 &\sst z^{-2}%
\end{bmatrix}
\hspace{-4pt}
\begin{bmatrix}
\sst  (1 + 5z^{-1})/2 &\,\sst  -2\\
\sst -1/2 &\,\sst 0
\end{bmatrix}
=
\mathbf{V}_1(z)\bsy{\Delta}_1(z)\mathbf{Q}_2(z).
\end{align}

Factor  a diagonal gain matrix and a swap  off of $\mathbf{Q}_2(z)$,
\begin{align}\label{CDF75_col1_CCA_step2_result}
\mathbf{Q}_2(z)
&=
\begin{bmatrix}
\sst  -2 &\,\sst  0\\
\sst 0 &\,\sst -1/2
\end{bmatrix}
\hspace{-4pt}
\begin{bmatrix}
\sst  1 &\,\sst  -(1 + 5z^{-1})/4\\
\sst 0 &\sst 1
\end{bmatrix}
\hspace{-4pt}
\begin{bmatrix}
\sst  0 &\,\sst 1 \\
\sst 1 &\,\sst 0
\end{bmatrix}
= \mathbf{D}_{-2,-1/2}\mathbf{V}_2(z)\mathbf{J}.
\end{align}
Combine~\eqref{CDF75_col1_CCA_step0_result}, \eqref{CDF75_col1_CCA_step1_result}, and~\eqref{CDF75_col1_CCA_step2_result} to get the same factorization as~\eqref{CDF75_EEA_col1},
\begin{align*}
\mathbf{H}(z)
&=
\mathbf{V}_0(z)\,\mathbf{V}_1(z)\,\bsy{\Delta}_1(z)\,\mathbf{D}_{-2,-1/2}\mathbf{V}_2(z)\,\mathbf{J}\\
&=
-\mathbf{D}_{2,1/2}\,\bigl(\lowergam{2,1/2}\!\mathbf{V}_0(z)\bigr)\bigl(\lowergam{2,1/2}\!\mathbf{V}_1(z)\bigr)\bsy{\Delta}_1(z)\,\mathbf{V}_2(z)\,\mathbf{J},\text{ using~\eqref{symbolic_DIF},}\\
&=
\begin{bmatrix}
\sst -2 &\sst  0\\
\sst 0 &\sst  -1/2
\end{bmatrix}
\hspace{-4pt}
\begin{bmatrix}
\sst 1 &\sst  (-13 + 3z^{-1})/16\\
\sst 0 &\sst  1
\end{bmatrix}
\hspace{-4pt}
\begin{bmatrix}
\sst 1 &\sst  0\\
\sst 1 + z^{-1} &\sst 1
\end{bmatrix}
\hspace{-4pt}
\begin{bmatrix}
\sst 1 & \sst 0\\
\sst 0 &\sst z^{-2}
\end{bmatrix}
\hspace{-4pt}
\begin{bmatrix}
\sst 1 &\sst  -(1 + 5z^{-1})/4\\
\sst 0 &\sst  1
\end{bmatrix}
\hspace{-4pt}
\begin{bmatrix}
\sst 0 &\sst  1\\
\sst 1 &\sst  0
\end{bmatrix}.
\end{align*}
%


\subsubsection{Enumerating Degree-Lifting Factorizations}\label{sec:CubicBSpline:CCA_Col1:Unique} 
Theorem~\ref{thm:LDDRT} allows us, in principle, to enumerate  all possible degree-lifting decompositions of a given filter bank. 
E.g., what are the possible degree-reducing causal complements to $(F_0,F_1)$ for inhomogeneity $\hat{a}z^{-\hat{d}} = -z^{-2}$ in \eqref{CDF75_col1_CCA_step0_form}?  Since
%
\begin{align*}
\deg|\mathbf{Q}_1| = 
\deg\begin{vmatrix}
R_0 & R_1\\
F_0 & F_1\vspace{-1pt}%
\end{vmatrix}
&= \deg|\mathbf{Q}_0| 
= 2
< \deg(F_0) + \deg(F_1) 
= 3,
\end{align*}
Theorem~\ref{thm:LDDRT}\eqref{DegRed_3} predicts that  division in column~0 of $\mathbf{Q}_0$ will yield the same causal complement  $(R_0,R_1)$ (the top row of $\mathbf{Q}_1$ in~\eqref{CDF75_col1_CCA_step0_result}) that was obtained by division in column~1 of $\mathbf{Q}_0$.  This is confirmed by noting that the remainders  in the top row of $\mathbf{Q}_1$ are  degree-reducing in \emph{both} $F_0$ and $F_1$ and are therefore, by Theorem~\ref{thm:LDDRT}(\ref{DegRed_1}--\ref{DegRed_2}), the unique remainders given by division in \emph{either} column.

Factoring $\mathbf{Q}_1(z)$ as in~\eqref{CDF75_col1_CCA_step1_form} is another matter, however, because
\begin{align*}
\deg\begin{vmatrix}
F_0 & F_1\\
R_0 & R_1\vspace{-1pt}%
\end{vmatrix}
=
\deg|\mathbf{Q}_1| = 2 > \deg(F_0) + \deg(F_1)  = 1. 
\end{align*}
By Theorem~\ref{thm:LDDRT}\eqref{DegRed_3} the causal complements that are degree-reducing in $F_0$ and $F_1$ are \emph{different}. Division in column~1 of  $\mathbf{Q}_1(z)$ produces~\eqref{CDF75_col1_CCA_step1_result}, which leads to~\eqref{CDF75_EEA_col1}, so dividing in column~0  produces a different  result than~\eqref{CDF75_EEA_col1}.  There is no  EEA analogue of ``switching columns'' like this. 

This notion of enumerating degree-lifting decompositions of a given filter bank is pursued in~\cite[\S 2]{Bris:24:Causal-Complementation-Algorithm} where we introduce a bookkeeping schema for keeping track of the options for forming different lifting factorizations. This schema is incorporated into the  formulation of the Causal Complementation Algorithm~\cite[Algorithm~1]{Bris:24:Causal-Complementation-Algorithm}. A provably complete example using the Daubechies 4-tap/4-tap paraunitary filter bank is presented~\cite[Table~2]{Bris:24:Causal-Complementation-Algorithm} that enumerates all ``left degree-lifting'' factorizations (i.e., factorizations constrained by  partial pivoting to employ only elementary row reductions).

\subsection{CCA Factorization Using the SGDA in Column~1}\label{sec:CubicBSpline:SGDA}  
Another serious limitation of the causal EEA method is that  it  puts the entire determinant  in a \emph{single} diagonal delay matrix, e.g.,     $\diag(1, z^{-2})$  in~\eqref{CDF75_EEA_col1}. The CCA  allows one to factor out delays at arbitrary points in the factorization  using the SGDA, which  lets us construct CCA factorizations  not obtainable using the causal EEA.  We now derive a \emph{causal} version  of the linear phase factorization~\eqref{CDF75_PWA_anal}; as in Section~\ref{sec:CubicBSpline:CCA_Col1} we seek a factorization of the form~\eqref{CDF75_col1_CCA_step0_form},
\begin{align}
\begin{bmatrix}
\sst  (3 + 5z^{-1} + 5z^{-2} + 3z^{-3})/32 &\,\sst  (-3 + 10z^{-1} - 3z^{-2})/8\\
\sst (1 + 6z^{-1} + z^{-2})/8 &\,\sst  -(1 + z^{-1})/2 
\end{bmatrix} \label{CDF75_again}
=
\begin{bmatrix}
1 & S\\
0 & 1\vspace{-1pt}%
\end{bmatrix}
\hspace{-3pt}
\begin{bmatrix}
R_0 & R_1\\
F_0 & F_1\vspace{-1pt}%
\end{bmatrix}.
\end{align}
This time, however, dividing $F_1$ into $E_1$ using  the SGDA  with $M=1$ yields
\begin{align}\label{CDF75_col1_SGDA_first_div}
S(z) = 3(1 + z^{-1})/4 \text{\ \ with\ \ } R_1= 2z^{-1},\; \deg(R_1) < \deg(F_1) + 1.  
\end{align}
Set $R_0\leftarrow E_0 - F_0S  = -z^{-1}(1+z^{-1})/2$ and note that $z^{-1}$ also divides  $R_0(z)$ 
(this follows from~\cite[Theorem~2.3]{Bris:24:Causal-Complementation-Algorithm}).
The first lifting step can  be written
\begin{align}\label{CDF75_col1_SGDA_step0}
\mathbf{Q}_0(z)
&=
\begin{bmatrix}
\sst  1 &\sst   3(1 + z^{-1})/4\\
\sst  0 &\sst   1\vspace{-1pt}%
\end{bmatrix}
\hspace{-4pt}
\begin{bmatrix}
\sst  z^{-1} &\sst   0\\
\sst  0 &\sst   1\vspace{-1pt}%
\end{bmatrix}
\hspace{-4pt}
\begin{bmatrix}
\sst  -(1 + z^{-1})/2 &\sst  2\\
\sst (1 + 6z^{-1} + z^{-2})/8 &\,\sst  -(1 + z^{-1})/2 
\end{bmatrix}\\
&=
\mathbf{V}_0(z)\bsy{\Delta}_0(z)\mathbf{Q}_1(z). \nonumber
\end{align}
%

Next, seek a factorization of $\mathbf{Q}_1(z)$ of the  form 
\begin{align}
\mathbf{Q}_1(z)
&=
\begin{bmatrix}
\sst  -(1 + z^{-1})/2 &\,\sst  2\\
\sst (1 + 6z^{-1} + z^{-2})/8 &\,\sst  -(1 + z^{-1})/2 
\end{bmatrix}
=
\begin{bmatrix}
F_0 & F_1\\
E_0 & E_1\vspace{-1pt}
\end{bmatrix}
=
\begin{bmatrix}
1 & 0\\
S & 1\vspace{-1pt}%
\end{bmatrix}
\hspace{-3pt}
\begin{bmatrix}
F_0 & F_1\\
R_0 & R_1\vspace{-1pt}%
\end{bmatrix}.\label{CDF75_col1_SGDA_step1_form}
\end{align}
Dividing in column~1 using the classical polynomial division algorithm produces
$S(z) = -(1 + z^{-1})/4$ and $R_1=0$.
Set $R_0\leftarrow E_0 - F_0S = z^{-1}/2$; both $R_0$ and $R_1$ are divisible by $z^{-1}$ so the second lifting step can be written 
\begin{align}\label{CDF75_col1_SGDA_step1}
\mathbf{Q}_1(z)
&=
\begin{bmatrix}
\sst  1 &\,\sst   0\\
\sst  -(1 + z^{-1})/4 &\,\sst   1
\end{bmatrix}
\hspace{-4pt}
\begin{bmatrix}
\sst  1 &\sst   0\\
\sst  0 &\sst   z^{-1}
\end{bmatrix}
\hspace{-4pt}
\begin{bmatrix}
\sst  -(1 + z^{-1})/2 &\,\sst  2\\
\sst 1/2 &\,\sst  0 
\end{bmatrix}
=
\mathbf{V}_1(z)\bsy{\Delta}_1(z)\mathbf{Q}_2(z).
\end{align}

Factor a diagonal gain matrix and a swap matrix off of $\mathbf{Q}_2(z)$ to get
\begin{align}\label{CDF75_col1_SGDA_step2}
\mathbf{Q}_2(z)
&=
\begin{bmatrix}
\sst 2 &\sst  0\\
\sst 0 &\sst  1/2
\end{bmatrix}
\hspace{-4pt}
\begin{bmatrix}
\sst 1 &\sst  -(1 + z^{-1})/4\\
\sst 0 &\sst  1
\end{bmatrix}
\hspace{-4pt}
\begin{bmatrix}
\sst 0 &\sst  1\\
\sst 1 &\sst  0
\end{bmatrix}
=
\mathbf{D}_{2,1/2}\mathbf{V}_2(z)\,\mathbf{J}.
\end{align}
Combine~\eqref{CDF75_col1_SGDA_step0}, \eqref{CDF75_col1_SGDA_step1}, and~\eqref{CDF75_col1_SGDA_step2}
and put  in  standard causal lifting form~\eqref{std_causal_form}, 
\begin{align}
\mathbf{H}(z) 
&=
\mathbf{V}_0(z)\,\bsy{\Delta}_0(z)\,\mathbf{V}_1(z)\,\bsy{\Delta}_1(z)\,\mathbf{D}_{2,1/2}\,\mathbf{V}_2(z)\,\mathbf{J}\nonumber\\
&=
\mathbf{D}_{2,1/2}\,\mathbf{U}_{2}(z)\,\bsy{\Lambda}_{2}(z)\,\mathbf{U}_{1}(z)\,\bsy{\Lambda}_1(z)\,\mathbf{U}_0(z)\,\mathbf{J}\text{\ \ using~\eqref{symbolic_DIF},}\nonumber\\
&=
\begin{bmatrix}
\sst 2 &\sst  0\\
\sst 0 &\sst  1/2
\end{bmatrix}
\hspace{-4pt}
\begin{bmatrix}
\sst 1 &\sst  3(1 + z^{-1})/16\\
\sst 0 &\sst  1
\end{bmatrix}
\hspace{-4pt}
\begin{bmatrix}
\sst  z^{-1} &\sst   0\\
\sst  0 &\sst   1
\end{bmatrix}
\hspace{-4pt}
\begin{bmatrix}
\sst 1 &\sst  0\\
\sst -(1 + z^{-1}) &\sst  1
\end{bmatrix}
\hspace{-4pt}
\begin{bmatrix}
\sst  1 &\sst   0\\
\sst  0 &\sst   z^{-1}
\end{bmatrix}
\hspace{-4pt}
\begin{bmatrix}
\sst 1 &\sst  -(1 + z^{-1})/4\\
\sst 0 &\sst  1
\end{bmatrix}
\hspace{-4pt}
\begin{bmatrix}
\sst  0 &\sst   1\\
\sst  1 &\sst   0
\end{bmatrix}.\label{CDF75_PWD_WSGLS_anal_lifting}
\end{align}

This is a causal analogue of the unimodular  WS group  lifting factorization~\eqref{CDF75_PWA_anal}; its causal linear phase lifting filters differ from those in~\eqref{CDF75_PWA_anal} by at most  delays. The  reader can confirm that \eqref{CDF75_PWD_WSGLS_anal_lifting} is also obtained using the  SGDA  with $M=1$ in row~1 of~\eqref{CDF75}  when computing the  lifting  by column reduction.
In comparison, \eqref{CDF75_PWD_WSGLS_anal_lifting} is \emph{not} produced by  the causal EEA method in \emph{any} row or column of~\eqref{CDF75}.  One might wonder what happens if one uses the SGDA with $M=1$ (as in~\eqref{CDF75_col1_SGDA_first_div}) to perform the first division operation in the EEA calculation in Section~\ref{sec:CubicBSpline:EEA_Col1}. This cannot possibly produce~\eqref{CDF75_PWD_WSGLS_anal_lifting}, however, because the augmentation matrix will still put the entire determinantal delay ($z^{-2}$) into a single diagonal delay matrix. Indeed, instead of~\eqref{CDF75_EEA_col1} or~\eqref{CDF75_PWD_WSGLS_anal_lifting} the EEA method using the SGDA in this manner yields the factorization
%
%
\begin{align}
\mathbf{H}(z) 
&=
\begin{bmatrix}
\sst 2 &\sst  0\\
\sst 0 &\sst  -1/2
\end{bmatrix}
\hspace{-4pt}
\begin{bmatrix}
\sst 1 &\sst  -3(1 + z^{-1})/16\\
\sst 0 &\sst  1
\end{bmatrix}
\hspace{-4pt}
\begin{bmatrix}
\sst 1 &\sst  0\\
\sst 1 &\sst 1
\end{bmatrix}
\hspace{-4pt}
\begin{bmatrix}
\sst 1 &\sst z^{-1}\\
\sst 0 &\sst  1
\end{bmatrix}
\hspace{-4pt}
\begin{bmatrix}
\sst z^{-2} & \sst 0\\
\sst 0 &\sst 1
\end{bmatrix}
\hspace{-4pt}
\begin{bmatrix}
\sst 1 &\sst 0\\
\sst  -(1 + 5z^{-1})/4 &\sst  1
\end{bmatrix} .\label{CDF75_EEA_col1_SGDA}
\end{align}
Since both  use the SGDA division in~\eqref{CDF75_col1_SGDA_first_div}, formulas~\eqref{CDF75_PWD_WSGLS_anal_lifting} and~\eqref{CDF75_EEA_col1_SGDA} have the same leftmost lifting filter, modulo a sign difference attributable to using~\eqref{symbolic_DIF} with the diagonal matrix in~\eqref{CDF75_EEA_col1_SGDA}.

\section{Conclusions}\label{sec:Conclusions}
We  introduce  a new   method,  the \emph{Causal Complementation Algorithm} (CCA), for factoring causal FIR perfect reconstruction transfer matrices into causal lifting steps (elementary matrices). 
Daubechies and Sweldens~\cite{DaubSwel98} factored \emph{unimodular}  filter banks using the Extended Euclidean Algorithm (EEA) for Laurent polynomials, but the CCA is more general, using Gaussian elimination to generate more causal lifting factorizations than those produced using the causal version of the EEA method. This includes a \hyperref[CDF75_PWD_WSGLS_anal_lifting]{causal linear phase analogue} of a  \hyperref[CDF75_PWA_anal]{unimodular linear phase  lifting factorization} for the CDF(7,5) cubic B-spline wavelet filter bank that is not generated by the causal EEA approach.

The causal EEA  works  within a single row or column of $\mathbf{H}(z)$, effectively limiting the EEA to constructing at most four lifting factorizations of a given filter bank (using classical polynomial division), no matter how big the matrix polynomial order of $\mathbf{H}(z)$ is.
The EEA method  constructs an augmentation matrix that is in turn lifted to $\mathbf{H}(z)$,  further limiting possible factorizations by putting the entire determinant of  $\mathbf{H}(z)$ into a \hyperref[CDF75_EEA_col1]{single diagonal delay matrix}, which can  lead to ill-conditioned factorizations~\cite{Bris:24:Causal-Complementation-Algorithm}.  
The CCA, in comparison, factors the whole matrix $\mathbf{H}(z)$ from the outset and enables strategies like \hyperref[sec:Study:Other]{switching the row or column in which division takes place} that have no EEA analogue.
The CCA allows the user to factor off diagonal delay matrices at will using the \hyperref[alg:SGDA]{Slightly Generalized Division Algorithm},  expanding the choice of lifting factorizations that can be obtained.

A  \hyperref[deg_add_bound]{max-additive inequality} connects causality to uniqueness of degree-reducing solutions to  linear Diophantine equations over causal \emph{polynomial} rings via a \hyperref[thm:LDDRT]{Linear Diophantine Degree-Reduction Theorem}; this inequality fails for \emph{Laurent} polynomials.  The LDDRT  uses the determinantal degree (information that is missing from unimodular normalizations) to decide whether a  causal  lifting step can be factored off in exactly one or two distinct ways. This allows users in principle to systematically generate all possible  \emph{degree-lifting} factorizations of a given causal filter bank, a capability not provided by the EEA approach. The author maintains that the degree-reducing properties of such factorizations should distinguish \emph{lifting} factorizations amongst all elementary matrix decompositions of polyphase matrices, a definition not made in~\cite{DaubSwel98}.

Work in progress  includes  specification of the CCA in algorithmic form,  complexity analysis showing the  advantages of the CCA over the EEA~\cite{Bris:24:Causal-Complementation-Algorithm}, specializations to linear phase filter banks, and realization theory for causal lifting.

\bibliographystyle{elsarticle-num}
\bibliography{CMBstring,CMBpubs,acad-press,elsevier,IEEE,Math-Soc,Misc,Oxbridge,prentice-hall,springer,standards,Theses,CMBcrossref}

\begin{thebibliography}{10}
\expandafter\ifx\csname url\endcsname\relax
  \def\url#1{\texttt{#1}}\fi
\expandafter\ifx\csname urlprefix\endcsname\relax\def\urlprefix{URL }\fi
\expandafter\ifx\csname href\endcsname\relax
  \def\href#1#2{#2} \def\path#1{#1}\fi

\bibitem{Mallat89c}
S.~G. Mallat, Multiresolution approximations and wavelet orthonormal bases of
  {L}$^2$({R}), Transactions of the Amer.\ Math.\ Soc. 315~(1) (1989) 69--87.

\bibitem{Daub92}
I.~C. Daubechies, Ten Lectures on Wavelets, no.~61 in {CBMS-NSF} Regional
  Conf.\ Series in Appl.\ Math., (Univ.\ Mass.---Lowell, June 1990),
  Soc.~Indust.\ Appl.\ Math., Philadelphia, 1992.

\bibitem{Meyer93}
Y.~Meyer, Wavelets: Algorithms and Applications, Soc.~Indust.\ Appl.\ Math.,
  Philadelphia, 1993.

\bibitem{Mallat99}
S.~Mallat, A Wavelet Tour of Signal Processing, 2nd Edition, Academic Press,
  San Diego, CA, 1999.

\bibitem{Bris:13b:TIT}
C.~M. Brislawn, Group-theoretic structure of linear phase multirate filter
  banks, IEEE Trans.\ Information Theory 59~(9) (2013) 5842--5859.
\newblock \href {http://arxiv.org/abs/1309.7665} {\path{arXiv:1309.7665}},
  \href {https://doi.org/10.1109/TIT.2013.2259292}
  {\path{doi:10.1109/TIT.2013.2259292}}.

\bibitem{Vaid93}
P.~P. Vaidyanathan, Multirate Systems and Filter Banks, Prentice Hall,
  Englewood Cliffs, NJ, 1993.

\bibitem{VettKov95}
M.~Vetterli, J.~Kova\v{c}evi\'{c}, Wavelets and Subband Coding, Prentice Hall,
  Englewood Cliffs, NJ, 1995.

\bibitem{StrNgu96}
G.~Strang, T.~Nguyen, Wavelets and Filter Banks, Wellesley-Cambridge,
  Wellesley, MA, 1996.

\bibitem{TranQueirozNguyen:00:Linear-phase-perfect-reconstruction}
T.~D. Tran, R.~L. de~Queiroz, T.~Q. Nguyen, Linear-phase perfect reconstruction
  filter bank: {L}attice structure, design, and application in image coding,
  IEEE Trans.\ Signal Process. 48~(1) (2000) 133--147.

\bibitem{GaoNguyenStrang:01:M-channel-PUFBs}
X.~Gao, T.~Q. Nguyen, G.~Strang, On factorization of {M}-channel paraunitary
  filterbanks, IEEE Trans.\ Signal Process. 49~(7) (2001) 1433--1446.

\bibitem{OrainTranHellerNgu:01:paraunitary-linear-phase}
S.~Oraintara, T.~D. Tran, P.~N. Heller, T.~Q. Nguyen, Lattice structure for
  regular paraunitary linear-phase filterbanks and {M}-band orthogonal
  symmetric wavelets, IEEE Transactions on Signal Processing 49~(11) (2001)
  2659--2672.

\bibitem{GanMaNguyenEtal:02:on-completeness}
L.~Gan, K.-K. Ma, T.~Nguyen, T.~Tran, R.~de~Queiroz, On the completeness of the
  lattice factorization for linear-phase perfect reconstruction filter banks,
  Signal Processing Letters, IEEE 9~(4) (2002) 133--136.
\newblock \href {https://doi.org/10.1109/LSP.2002.1001651}
  {\path{doi:10.1109/LSP.2002.1001651}}.

\bibitem{OrainTranNgu:03:regular-LPFBs}
S.~Oraintara, T.~D. Tran, T.~Q. Nguyen, A class of regular biorthogonal
  linear-phase filterbanks: {T}heory, structure, and application in image
  coding, IEEE Trans.\ Signal Process. 51~(12) (2003) 3220--3235.

\bibitem{GanMa:TCAS-II-04:simplified-lattice}
L.~Gan, K.-K. Ma, A simplified lattice factorization for linear-phase
  paraunitary filter banks with pairwise mirror image frequency responses,
  Circuits and Systems II: Express Briefs, IEEE Transactions on 51~(1) (2004)
  3--7.
\newblock \href {https://doi.org/10.1109/TCSII.2003.821515}
  {\path{doi:10.1109/TCSII.2003.821515}}.

\bibitem{GanMa:TSP-04:simplified-order-one}
L.~Gan, K.-K. Ma, On simplified order-one factorizations of paraunitary
  filterbanks, Signal Processing, IEEE Transactions on 52~(3) (2004) 674--686.
\newblock \href {https://doi.org/10.1109/TSP.2003.822356}
  {\path{doi:10.1109/TSP.2003.822356}}.

\bibitem{MakMuthRed:04:Eigenstructure-approach}
A.~Makur, A.~Muthuvel, P.~Reddy, Eigenstructure approach for complete
  characterization of linear-phase {FIR} perfect reconstruction analysis length
  2{M} filterbanks, Signal Processing, IEEE Transactions on 52~(6) (2004) 1801
  -- 1804.
\newblock \href {https://doi.org/10.1109/TSP.2004.827201}
  {\path{doi:10.1109/TSP.2004.827201}}.

\bibitem{XuMakur:09:Arbitrary-Length-LPPRFB}
Z.~Xu, A.~Makur, On the arbitrary-length {M}-channel linear phase perfect
  reconstruction filter banks, Signal Processing, IEEE Transactions on 57~(10)
  (2009) 4118--4123.
\newblock \href {https://doi.org/10.1109/TSP.2009.2024026}
  {\path{doi:10.1109/TSP.2009.2024026}}.

\bibitem{PainterSpanias:00:Perceptual-Coding}
T.~Painter, A.~Spanias, Perceptual coding of digital audio, Proceedings of the
  IEEE 88~(4) (2000) 451--513.

\bibitem{KhaTuanNgu:ICASSP07:SDP-Cos-Mod-FB}
H.~Kha, H.~Tuan, T.~Nguyen, An efficient {SDP} based design for prototype
  filters of {M}-channel cosine-modulated filter banks, in: Acoustics, Speech
  and Signal Processing, 2007. ICASSP 2007. IEEE International Conference on,
  Vol.~3, 2007, pp. 893--896.
\newblock \href {https://doi.org/10.1109/ICASSP.2007.366824}
  {\path{doi:10.1109/ICASSP.2007.366824}}.

\bibitem{KhaTuanNguyen:09:Cos-Mod-FBs}
H.~H. Kha, H.~D. Tuan, T.~Nguyen, Efficient design of cosine-modulated filter
  banks via convex optimization, Signal Processing, IEEE Transactions on 57~(3)
  (2009) 966--976.
\newblock \href {https://doi.org/10.1109/TSP.2008.2009268}
  {\path{doi:10.1109/TSP.2008.2009268}}.

\bibitem{Sweldens96}
W.~Sweldens, The lifting scheme: a custom-design construction of biorthogonal
  wavelets, Appl.\ Comput.\ Harmonic Anal. 3~(2) (1996) 186--200.

\bibitem{Sweldens:98:SIAM-lifting-scheme}
W.~Sweldens, The lifting scheme: a construction of second generation wavelets,
  SIAM J.~Math.\ Anal. 29~(2) (1998) 511--546.

\bibitem{DaubSwel98}
I.~Daubechies, W.~Sweldens, Factoring wavelet transforms into lifting steps,
  J.~Fourier Analysis \& Appl. 4~(3) (1998) 247--269.
\newblock \href {https://doi.org/10.1007/BF02476026}
  {\path{doi:10.1007/BF02476026}}.

\bibitem{ISO_15444_1}
ISO/IEC Int'l.\ Standard~15444-1, ITU-T Rec.~T.800,
  \href{http://www.itu.int/rec/T-REC-T.800/en}{Information
  technology---{JPEG}~2000 image coding system, {P}art~1}, Int'l.~Org.\
  Standardization, 2000.
\newline\urlprefix\url{http://www.itu.int/rec/T-REC-T.800/en}

\bibitem{ISO_15444_2}
ISO/IEC Int'l.\ Standard~15444-2, ITU-T Rec.~T.801,
  \href{http://www.itu.int/rec/T-REC-T.801/en}{Information
  technology---{JPEG}~2000 image coding system, {P}art~2: {E}xtensions},
  Int'l.~Org.\ Standardization, 2004.
\newline\urlprefix\url{http://www.itu.int/rec/T-REC-T.801/en}

\bibitem{TaubMarc02}
D.~S. Taubman, M.~W. Marcellin, JPEG2000: Image Compression Fundamentals,
  Standards, and Practice, Kluwer, Boston, MA, 2002.

\bibitem{BrisQuirk03}
C.~M. Brislawn, M.~D. Quirk, Image compression with the {JPEG-2000} standard,
  in: R.~G. Driggers (Ed.), {E}ncyclopedia of {O}ptical {E}ngineering, Marcel
  Dekker, New York, 2003, pp. 780--785, invited book chapter.

\bibitem{AcharyaTsai:04:JPEG2000-Standard}
T.~Acharya, P.-S. Tsai, {JPEG2000} Standard for Image Compression: {C}oncepts,
  Algorithms and {VLSI} Architectures, Wiley-Interscience, 2004.

\bibitem{Lee:05:J2K-Retrospective}
D.~T. Lee, {JPEG~2000}: {R}etrospective and new developments, Proceedings of
  the IEEE 93~(1) (2005) 32--41.

\bibitem{CCSDS_122.0:2005}
Recommendation for Space Data System Standards, Consultative Committee for
  Space Data Systems, \href{https://public.ccsds.org/Pubs/122x0b2e1.pdf}{Image
  Data Compression}, no. CCSDS 122.0-B-2, NASA, Washington, D.C., 2017.
\newline\urlprefix\url{https://public.ccsds.org/Pubs/122x0b2e1.pdf}

\bibitem{YehArmKielyEtal:AC-05:CCSDS-image-comp}
P.-S. Yeh, P.~Armbruster, A.~Kiely, B.~Masschelein, G.~Moury, C.~Schaefer,
  C.~Thiebaut, The new {CCSDS} image compression recommendation, in: Proc.\
  IEEE Aerospace Conf., 2005, pp. 4138--4145.
\newblock \href {https://doi.org/10.1109/AERO.2005.1559719}
  {\path{doi:10.1109/AERO.2005.1559719}}.

\bibitem{MacLaneBirkhoff67}
S.~MacLane, G.~Birkhoff, Algebra, Macmillan, New York, NY, 1967.

\bibitem{BachShallit:96:Algo-Number-Theory}
E.~Bach, J.~Shallit, Algorithmic Number Theory, Vol.~1 of Foundations of
  Computing, MIT Press, Cambridge, MA, 1996.

\bibitem{Shoup:05:Comp-Number-Theory}
V.~Shoup, A Computational Introduction to Number Theory and Algebra, Cambridge
  Univ.\ Press, Cambridge, UK, 2005.

\bibitem{GathenGerhard:13:Modern-Computer-Algebra}
J.~von~zur Gathen, J.~Gerhard, Modern Computer Algebra, 3rd Edition, Cambridge
  Univ.\ Press, Cambridge, UK, 2013.
\newblock \href {https://doi.org/10.1017/CBO9781139856065}
  {\path{doi:10.1017/CBO9781139856065}}.

\bibitem{Herley93}
C.~Herley, Wavelets and filter banks, Ph.D. thesis, Columbia University, New
  York, NY (1993).

\bibitem{VetterliHerley:92:Wavelets-filter-banks}
M.~Vetterli, C.~Herley,
  \href{https://infoscience.epfl.ch/record/33904/files/VetterliH92.pdf}{Wavelets
  and filter banks: theory and design}, IEEE Trans.\ Signal Process. 40~(9)
  (1992) 2207--2232.
\newblock \href {https://doi.org/10.1109/78.157221}
  {\path{doi:10.1109/78.157221}}.
\newline\urlprefix\url{https://infoscience.epfl.ch/record/33904/files/VetterliH92.pdf}

\bibitem{TolHollKal:95:realiz-biorthog-M-D}
L.~Tolhuizen, H.~Hollmann, T.~Kalker, On the realizability of biorthogonal,
  {M}-dimensional two-band filter banks, IEEE Transactions on Signal Processing
  43~(3) (1995) 640--648.

\bibitem{Cohn:66:structure-ring}
P.~M. Cohn, \href{http://www.numdam.org/item?id=PMIHES_1966__30__5_0}{On the
  structure of the {$GL_2$} of a ring}, Publications mathematiques de
  l'I.H.E.S. 30 (1966) 5--53.
\newline\urlprefix\url{http://www.numdam.org/item?id=PMIHES_1966__30__5_0}

\bibitem{Suslin:77:structure-special}
A.~A. Suslin, On the structure of the special linear group over polynomial
  rings, Math. USSR Izv. 11 (1977) 221--238.

\bibitem{ParkWoodburn:95:algorithmic-proof}
H.~Park, C.~Woodburn, An algorithmic proof of {S}uslin's stability theorem for
  polynomial rings, Journal of Algebra 178 (1995) 277--298.

\bibitem{AdamsKossentini:00:Reversible-wavelet-transforms}
M.~Adams, F.~Kossentini, Reversible integer-to-integer wavelet transforms for
  image compression: {P}erformance evaluation and analysis, Image Processing,
  IEEE Transactions on 9~(6) (2000) 1010--1024.
\newblock \href {https://doi.org/10.1109/83.846244}
  {\path{doi:10.1109/83.846244}}.

\bibitem{MaslenAbbott:00:Automation-lifting-factorisation}
M.~Maslen, P.~Abbott, Automation of the lifting factorisation of wavelet
  transforms, Computer Physics Communications 127~(2-3) (2000) 309--326.
\newblock \href {https://doi.org/10.1016/S0010-4655(99)00451-8}
  {\path{doi:10.1016/S0010-4655(99)00451-8}}.

\bibitem{ShuiBaoEtal:02:Two-channel-adaptive-biorthogonal}
P.~Shui, Z.~Bao, X.~Zhang, Y.~Tang, Two-channel adaptive biorthogonal
  filterbanks via lifting, Signal Process. 82~(6) (2002) 881--893.
\newblock \href
  {https://doi.org/http://dx.doi.org/10.1016/S0165-1684(02)00196-2}
  {\path{doi:http://dx.doi.org/10.1016/S0165-1684(02)00196-2}}.

\bibitem{AdamsWard:03:Symmetric-extension-compatible}
M.~D. Adams, R.~K. Ward, Symmetric-extension-compatible reversible
  integer-to-integer wavelet transforms, IEEE Trans.\ Signal Process. 51~(10)
  (2003) 2624--2636.

\bibitem{LiaoMandalEtal:04:Efficient-architectures-lifting-based}
H.~Liao, M.~K. Mandal, B.~F. Cockburn, Efficient architectures for 1-{D} and
  2-{D} lifting-based wavelet transforms, IEEE Trans.\ Signal Process. 52~(5)
  (2004) 1315--1326.

\bibitem{Tran:02:M-channel-linear-phase}
T.~D. Tran, M-channel linear phase perfect reconstruction filter bank with
  rational coefficients, IEEE Trans.\ Circuits Systems~I 49~(7) (2002)
  914--927.

\bibitem{Tran:02:Rational-LPPRFBs}
T.~Tran, M-channel linear phase perfect reconstruction filter bank with
  rational coefficients, Circuits and Systems I: Fundamental Theory and
  Applications, IEEE Transactions on 49~(7) (2002) 914 --927.
\newblock \href {https://doi.org/10.1109/TCSI.2002.800467}
  {\path{doi:10.1109/TCSI.2002.800467}}.

\bibitem{ChenAmaratun:03:M-channel-lifting-factorization}
Y.-J. Chen, K.~S. Amaratunga, M-channel lifting factorization of perfect
  reconstruction filter banks and reversible {M}-band wavelet transforms, IEEE
  Trans.\ Circuits Systems~II 50~(12) (2003) 963--976.

\bibitem{ShuiBao:04:M-band-biorthogonal-interpolating}
P.-L. Shui, Z.~Bao, M-band biorthogonal interpolating wavelets via lifting
  scheme, IEEE Trans.\ Signal Process. 52~(9) (2004) 2500--2512.

\bibitem{ChenOrainAmara:05:Dyadic-based-factorizations}
Y.-J. Chen, S.~Oraintara, K.~Amaratunga, Dyadic-based factorizations for
  regular paraunitary filterbanks and {M}-band orthogonal wavelets with
  structural vanishing moments, IEEE Transactions on Signal Processing 53~(1)
  (2005) 193--207.

\bibitem{IwamuraTanakaIkehara:07:Efficient-Lifting}
S.~Iwamura, Y.~Tanaka, M.~Ikehara, An efficient lifting structure of
  biorthogonal filter banks for lossless image coding, in: Image Processing,
  IEEE International Conference on, Vol.~6, 2007, pp. 433--436.
\newblock \href {https://doi.org/10.1109/ICIP.2007.4379614}
  {\path{doi:10.1109/ICIP.2007.4379614}}.

\bibitem{TanIkeNgu:08:LPFB-Lattice-Structure}
Y.~Tanaka, M.~Ikehara, T.~Nguyen, A lattice structure of biorthogonal
  linear-phase filter banks with higher order feasible building blocks,
  Circuits and Systems I: Regular Papers, IEEE Transactions on 55~(8) (2008)
  2322 --2331.
\newblock \href {https://doi.org/10.1109/TCSI.2008.918225}
  {\path{doi:10.1109/TCSI.2008.918225}}.

\bibitem{SuzukiIkeharaNguyen:12:Generalized-Block-Lifting}
T.~Suzuki, M.~Ikehara, T.~Q. Nguyen, Generalized block-lifting factorization of
  {M}-channel biorthogonal filter banks for lossy-to-lossless image coding,
  IEEE Transactions on Image Processing 21~(7) (2012) 3220--3228.
\newblock \href {https://doi.org/10.1109/TIP.2012.2190611}
  {\path{doi:10.1109/TIP.2012.2190611}}.

\bibitem{WengChenVaid:10:General-Triang-Decomp}
C.-C. Weng, C.-Y. Chen, P.~P. Vaidyanathan, Generalized triangular
  decomposition in transform coding, Signal Processing, IEEE Transactions on
  58~(2) (2010) 566--574.
\newblock \href {https://doi.org/10.1109/TSP.2009.2031733}
  {\path{doi:10.1109/TSP.2009.2031733}}.

\bibitem{WengVaid:12:GTD-Optimizing-PRFBs}
C.-C. Weng, P.~P. Vaidyanathan, The role of {GTD} in optimizing perfect
  reconstruction filter banks, Signal Processing, IEEE Transactions on 60~(1)
  (2012) 112 --128.
\newblock \href {https://doi.org/10.1109/TSP.2011.2169252}
  {\path{doi:10.1109/TSP.2011.2169252}}.

\bibitem{FooteMirchandEtal:00:Wreath-Product-Group}
R.~Foote, G.~Mirchandani, D.~Rockmore, D.~Healy, T.~Olson, A wreath product
  group approach to signal and image processing--{Part~I: M}ultiresolution
  analysis, IEEE Trans.\ Signal Process. 48~(1) (2000) 102--132.

\bibitem{MirchandFooteEtal:00:Wreath-Product-Group}
G.~Mirchandani, R.~Foote, D.~Rockmore, D.~Healy, T.~Olson, A wreath product
  group approach to signal and image processing--{Part~II: C}onvolution,
  correlation, and applications, IEEE Trans.\ Signal Process. 48~(3) (2000)
  749--767.

\bibitem{FooteMirchandEtal:04:Two-Dimensional-Wreath-Product}
R.~Foote, G.~Mirchandani, D.~Rockmore, Two-dimensional wreath product
  group-based image processing, J.~Symbolic Comput. 37~(2) (2004) 187--207.

\bibitem{Park:04:Symbolic-computation-signal}
H.~Park, Symbolic computation and signal processing, Journal of Symbolic
  Computation 37~(2) (2004) 209--226.
\newblock \href {https://doi.org/10.1016/j.jsc.2002.06.003}
  {\path{doi:10.1016/j.jsc.2002.06.003}}.

\bibitem{LebrunSelesnic:04:Grobner-bases-wavelet}
J.~Lebrun, I.~Selesnick, Gr{\"o}bner bases and wavelet design, J.~Symbolic
  Comput. 37~(2) (2004) 227--259.

\bibitem{DuBhosriFrazho:10:FB-commutant-lifting}
D.~Du, W.~Bhosri, A.~Frazho, Multirate filterbank design: A relaxed commutant
  lifting approach, Signal Processing, IEEE Transactions on 58~(4) (2010) 2102
  --2112.
\newblock \href {https://doi.org/10.1109/TSP.2010.2040686}
  {\path{doi:10.1109/TSP.2010.2040686}}.

\bibitem{HurParkZheng:14:Multi-D-Wavelet}
Y.~Hur, H.~Park, F.~Zheng, Multi-{D} wavelet filter bank design using
  {Quillen-Suslin} theorem for {L}aurent polynomials, Signal Processing, IEEE
  Transactions on 62~(20) (2014) 5348--5358.
\newblock \href {https://doi.org/10.1109/TSP.2014.2347263}
  {\path{doi:10.1109/TSP.2014.2347263}}.

\bibitem{BrisWohl06}
C.~M. Brislawn, B.~Wohlberg, The polyphase-with-advance representation and
  linear phase lifting factorizations, IEEE Trans.\ Signal Process. 54~(6)
  (2006) 2022--2034.
\newblock \href {https://doi.org/10.1109/TSP.2006.872582}
  {\path{doi:10.1109/TSP.2006.872582}}.

\bibitem{Bris:10:GLS-I}
C.~M. Brislawn, Group lifting structures for multirate filter banks~{I}:
  {U}niqueness of lifting factorizations, IEEE Trans.\ Signal Process. 58~(4)
  (2010) 2068--2077.
\newblock \href {http://arxiv.org/abs/1310.2206} {\path{arXiv:1310.2206}},
  \href {https://doi.org/10.1109/TSP.2009.2039816}
  {\path{doi:10.1109/TSP.2009.2039816}}.

\bibitem{Bris:10b:GLS-II}
C.~M. Brislawn, Group lifting structures for multirate filter banks~{II}:
  {L}inear phase filter banks, IEEE Trans.\ Signal Process. 58~(4) (2010)
  2078--2087.
\newblock \href {http://arxiv.org/abs/1310.2208} {\path{arXiv:1310.2208}},
  \href {https://doi.org/10.1109/TSP.2009.2039818}
  {\path{doi:10.1109/TSP.2009.2039818}}.

\bibitem{Bris:13:FFT}
C.~M. Brislawn, On the group-theoretic structure of lifted filter banks, in:
  T.~Andrews, R.~Balan, J.~Benedetto, W.~Czaja, K.~Okoudjou (Eds.), Excursions
  in Harmonic Analysis, {\rm vol.~2}, Applied and Numerical Harmonic Analysis,
  Birkh\"auser, Boston, 2013, pp. 113--135, invited book chapter.
\newblock \href {http://arxiv.org/abs/1310.0530} {\path{arXiv:1310.0530}},
  \href {https://doi.org/10.1007/978-0-8176-8379-5_6}
  {\path{doi:10.1007/978-0-8176-8379-5_6}}.

\bibitem{Bris:24:Causal-Complementation-Algorithm}
C.~M. Brislawn, The causal complementation algorithm for lifting factorization
  of perfect reconstruction multirate filter banks, \rm{submitted for
  publication} (June 2024).
\newblock \href {http://arxiv.org/abs/2408.07970} {\path{arXiv:2408.07970}}.

\bibitem{LeGallTabatabai:88:Subband-coding-digital}
D.~LeGall, A.~Tabatabai, Subband coding of digital images using symmetric short
  kernel filters and arithmetic coding techniques, in: Proc.\ Int'l.\ Conf.\
  Acoust., Speech, Signal Process., IEEE Signal Process.\ Soc., New York City,
  1988, pp. 761--764.

\bibitem{Waerden:83:Geometry-algebra-ancient}
B.~L. van~der Waerden, Geometry and Algebra in Ancient Civilizations,
  Springer-Verlag, Berlin, 1983.

\bibitem{Jacobson74}
N.~Jacobson, Basic Algebra, Vol.~1, Freeman, San Francisco, CA, 1974.

\bibitem{Hungerford74}
T.~W. Hungerford, Algebra, Springer-Verlag, New York, NY, 1974.

\bibitem{Herstein75}
I.~N. Herstein, Topics in Algebra, Xerox, Lexington, MA, 1975.

\bibitem{Blahut:87:Fast-Algorithms}
R.~Blahut, Fast Algorithms for Digital Signal Processing, Addison-Wesley,
  Reading, MA, 1987.

\bibitem{Coppel:06:Number-Theory}
W.~A. Coppel, Number Theory: {A}n Introduction to Mathematics, {Part~A},
  Springer, New York, NY, 2006.

\bibitem{CohenDaubFeau92}
A.~Cohen, I.~C. Daubechies, J.-C. Feauveau, Biorthogonal bases of compactly
  supported wavelets, Commun.\ Pure Appl.\ Math. 45 (1992) 485--560.

\end{thebibliography}

\end{document}